\documentclass[11pt]{article}
\usepackage[margin=1in]{geometry}                
\usepackage{graphicx}
\usepackage{amssymb}
\usepackage{epstopdf}
\usepackage{amsmath}
\usepackage{amsthm}
\usepackage{url}
\usepackage{authblk}
\usepackage{subcaption}

\usepackage{xr-hyper} 
\usepackage{hyperref}

\newtheorem{lemma}{Lemma}
\newtheorem{corollary}{Corollary}

\graphicspath{{./Figures/}}

\usepackage{setspace}
\doublespace

\graphicspath{{./Figures/}} 
\newcommand\myeq{\mathrel{\stackrel{\makebox[0pt]{\mbox{\normalfont\tiny def}}}{=}}}
\newcommand{\ci}{\mathrel{\text{\scalebox{1.07}{$\perp\mkern-10mu\perp$}}}}

\title{Combining controls can improve power in two-stage association studies}
\author[1]{James Liley}
\affil[1]{Department of Medicine, University of Cambridge, Addenbrooke's Hospital, Cambridge, CB2 0SP, UK. E: \href{mailto:ajl88@cam.ac.uk}{ajl88@cam.ac.uk}}
\date{}

\begin{document}
\maketitle
\clearpage

\section*{Abstract}

High dimensional case control studies are ubiquitous in the biological sciences, particularly genomics. To maximise power while constraining cost and to minimise type-1 error rates, researchers typically seek to replicate findings in a second experiment on independent cohorts before proceeding with further analyses.

This paper presents a method in which control (or case) samples from the discovery cohort are re-used in the replication study. The theoretical implications of this method are discussed and simulations used to compare performance against the standard method in a range of circumstances. In several common study designs, a shared-control method allows a substantial improvement in power while retaining type-1 error rate control. 

An important area of potential application arises when control samples are difficult to recruit or ascertain; for example in inter-disease comparisons, or studies on degenerative diseases.
Using similar methods, a procedure is proposed for `partial replication' using a new independent cohort consisting of only controls. This methods can be used to provide some validation of findings when a full replication procedure is not possible.

The new method has differing sensitivity to confounding in study cohorts compared to the standard procedure, which must be considered in its application. Type-1 error rates in these scenarios are analytically and empirically derived, and an online tool for comparing power and error rates is provided. 

Although careful consideration must be made of all necessary assumptions, this method can enable more efficient use of data in genome-wide association studies (GWAS) and other applications.

\clearpage

\section*{Introduction}
High-dimensional case-control studies have become a mainstay of investigation of pathophysiology in complex diseases and traits. An important part of their analysis is the process of replication~\cite{wason12}, in which the results of a high-dimensional study are used to inform the design of a second study at a subset of the original variables, with a joint analysis used to determine overall association. 

Replicating studies in this way has the advantage of increasing the effective study sample sizes without requiring measurement of all variables in all samples. It also serves to protect against false-positives due to systematic errors in the original datasets, by re-testing association in a second nominally independent dataset.

Replication has a significant cost, and can require large numbers of samples, especially when associated variables have small effects (ie \cite{fuchsberger16}). There is therefore a need to minimise the number of additional samples which need to be analysed. This paper presents a method to perform replication by combining controls in both the original `discovery' and second `replication' datasets, potentially reducing the number of new samples required. Shared-control approaches can improve study efficiency in many related applications in which studies are compared~\cite{lin09,han16,bhattacharjee12,zaykin10,liley15,fortune15}. 

Results from original and replication datasets for which some or all controls are shared cannot be directly compared due to the correlation between test statistics directly resulting from shared controls even under the null hypothesis~\cite{bhattacharjee12}; use of the same thresholds in a shared-control design as used in an independent-controls design will lead to higher type-1 error rates. 
This paper demonstrates a simple adaptation to a standard design to account for the changed correlation structure and retain control of type-1 error rate, only requiring a change to one p-value threshold. 

The action of sharing control samples results in a different spectrum of sensitivity to confounding in study groups. It necessitates a sacrifice of type-1 error rate control in variables affected by confounding in the discovery-phase control cohort, but improves type-1 error rate control in variables affected by confounding in the replication-phase control cohort. Performance is largely equivalent to an independent-controls design for variables affected by confounding in either case cohort. 

The new spectrum of false positive rates can be advantageous in circumstances where control samples in the replication cohort are less well-ascertained than those in the discovery cohort. This may be the case in studies on degenerative disease, where control ascertainment is generally uncertain, and population-sourced controls may be used for replication. The shared-control design can reduce power losses from mis-specified controls in the replication cohort, as well as reducing false-positive rates caused by confounding in the cohort. 

When used with shared cases instead of controls, this method can be adapted to a `partial replication' procedure where only a new control set is used. Although not equivalent to a full replication in an independent dataset, the procedure enables improvement in type-1 error rates and control over confounding. This is applicable in studies on rare traits, where all available samples need to be included in the discovery analysis for adequate power.

Throughout this paper, GWAS terminology will be used (SNPs, allele frequency, variants etc) although the method is applicable to any high-dimensional case control study. `Controls' will be considered to generally be samples unaffected by a disease or trait of interest, although the method can be applied with case/control labels swapped, or applied to comparisons between subgroups of a case group. Asymptotic analytical results are established where possible, but all type 1/type 2 error rates are readily tractable empirically to good accuracy given study sizes and proposed p-value thresholds, and a tool is provided to do this at \url{https://wallacegroup-liley.shinyapps.io/replication_shared/}.

\section*{Results}

\subsection*{Overview of method}

We assume a GWAS dataset of a set of cases $C_{1}$ and controls $C_{0}$ used in a `discovery' phase of a GWAS or similar study, and corresponding sets of cases and controls $C_{1}'$, $C_{0}'$ in the replication phase. We assume that $C_{0}$ and $C_{1}$ are genotyped at a set of SNPs $S$ and $C_{0}'$, $C_{1}'$ at a set $S' \subseteq S$. 

For each SNP we designate $\mu_{1}$, $\mu_{0}$, $\mu_{1}'$, $\mu_{0}'$ as the population minor allele frequency in the corresponding group, and $m_{1}$, $m_{0}$, $m_{1}'$, $m_{0}'$ as the observed allele frequency (so $E(m_{i})=\mu_{i}$). We designate two null hypotheses; $H_{0}^{\cup}:(\mu_{1}=\mu_{0}) \cup (\mu_{1}'=\mu_{0}')$ and $H_{0}^{\cap}:(\mu_{1}=\mu_{0} = \mu_{1}'=\mu_{0}')$, noting that $H_{0}^{\cup} \supseteq H_{0}^{\cap}$. In a typical conservative GWAS approach, we seek to test against $H_{0}^{\cup}$, since $\mu_{1} \neq \mu_{0}$ or $\mu_{1}' \neq \mu_{0}'$ may hold at non-disease associated SNPs due to confounding in the original or replication studies respectively. 

A typical two-stage genetic testing procedure~\cite{wason12}, which we will refer to as method A, 
begins by comparing genotypes of $C_{1}$ and $C_{0}$ at SNPs $S$ generating p-values $p_{d}$ (discovery). A subset $S'$ of SNPs reaching putative significance level $p_{d}<\alpha$ are genotyped in $C_{0}'$ and $C_{1}'$, with genotypes compared to generate p-values $p_{r}$ (replication stage). Finally, genotypes are compared between $C_{0} \cup C_{0}'$ and $C_{1} \cup C_{1}'$ at SNPs $S'$ to generate p-values $p_{m}$ (meta-analytic stage). SNPs are designated as `hits' if $p_{d}< \alpha,p_{r}<\beta,p_{m}<\gamma$ for some $\beta$, $\gamma$, and all effects have the same direction.
%

The main modification proposed in this paper, denoted as method B, differs at the replication stage in that $C_{1}'$ is compared with $C_{0} \cup C_{0}'$ at $S'$ instead of just $C_{0}'$. The p-values resulting from the modified replication stage are termed $p_{s}$, and the criterion to designate a hit changed to $p_{d}< \alpha,p_{s}<\beta^*,p_{m}<\gamma$, with all effects in the same direction. The threshold $\beta^*$ is chosen to conserve type-1 error rate between methods (see methods section, appendix~\ref{apx:zcor}).

A second modification, denoted method C, combines $C_{0}$ and $C_{0}'$ at both the discovery and replication phase. This is analagous to a situation in which only a single control cohort is available, and a choice must be made to split it between discovery and replication procedures or to use it for both. In this case, $C_{0} \cup C_{0}'$ is compared with $C_{1}$ at SNPs $S$ in the discovery phase to produce p-values $p_{c}$, then $C_{0} \cup C_{0}'$ is compared with $C_{1}'$ at SNPs $S'$ at the replication phase and compared with $C_{1} \cup C_{1}'$ at the meta-analytic stage to produce p-values $p_{s}$ and $p_{m}$ as before. A hit is determined by $p_{d}< \alpha,p_{s}<\beta^\perp,p_{m}<\gamma$, with all effects in the same direction. Again, $\beta^\perp$ is chosen to maintain the type-1 error rate between methods.

\begin{figure}
\includegraphics[width=\textwidth]{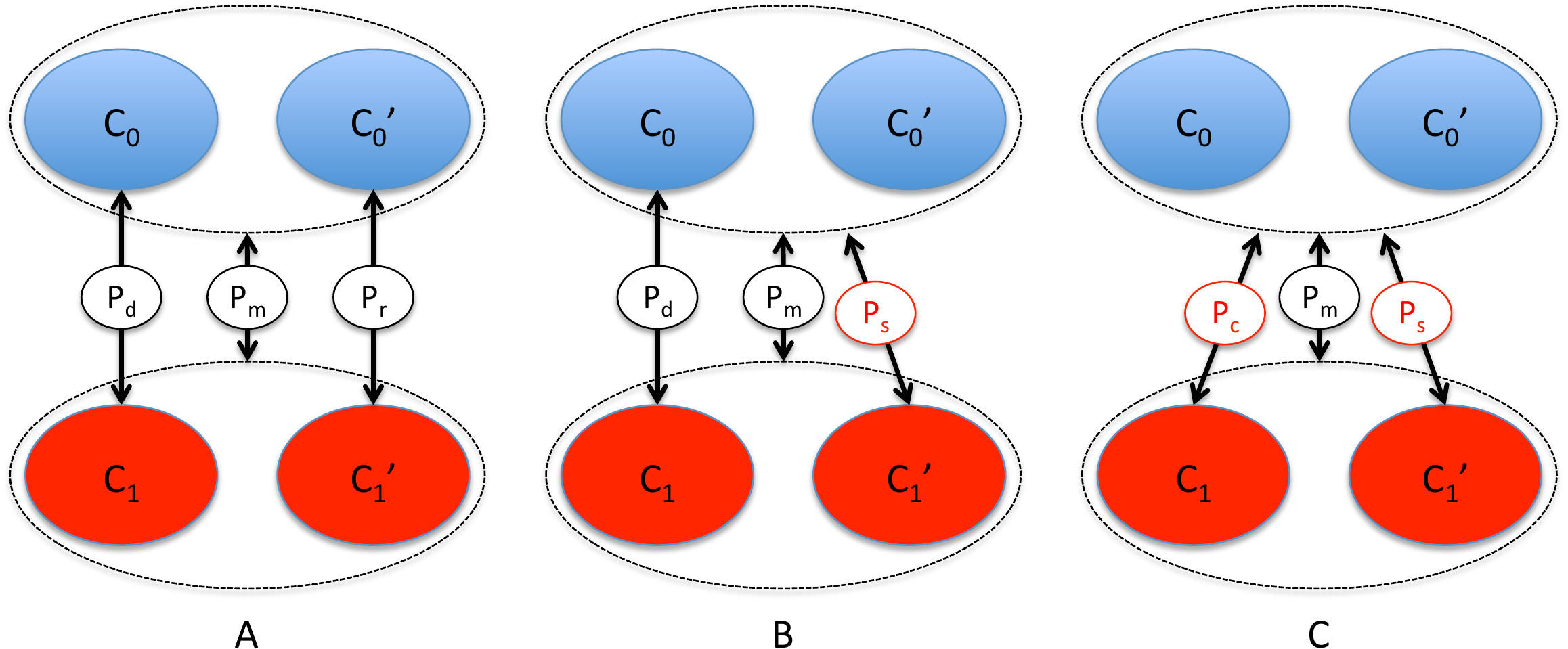}
\caption{Diagram of methods A, B, and C. Method B differs by comparing $C_{1}'$ to pooled $C_{0}$ and $C_{0}'$ at the replication stage to generate p-value $P_{s}$ instead of $P_{r}$. Method C also pools controls at the discovery phase, comparing $C_{1}$ to pooled $C_{0}$ and $C_{0}'$ to generate p-values $P_{c}$ instead of $P_{d}$. A `hit' is declared in method A if $P_{d}<\alpha$, $P_{r}<\beta$, $P_{m}<\gamma$, in method B if $P_{d}<\alpha$, $P_{s}<\beta^*$, $P_{m}<\gamma$ and in method C if $P_{c}<\alpha$, $P_{s}<\beta^\perp$, $P_{m}<\gamma$.
}
\label{fig:exp}
\end{figure}


%
%
\subsection*{General properties}

For SNPs in $H_{0}^{\cap}$, the overall type-1 error rate is conserved between methods by the definition of $\beta^*$, $\beta^{\perp}$ (equation~\ref{eq:betastardef}) at a level $P_{0}$. It is shown in appendix~\ref{apx:betabound} that $\beta>\beta^*>\beta^\perp$. 
For SNPs in $H_{0}^{\cup} \setminus H_{0}^{\cap}$ the type-1 error rates differ between methods. Such SNPs may be characterised by the group(s) amongst $C_{0}$, $C_{1}$, $C_{0}'$, $C_{1}'$ in which their expected MAF is aberrant from the expected MAF in the population which the group ostensibly represents. `Aberrance' is taken to mean an incorrect expected value from systematic measurement error or uncorrected confounding, rather than random deviance around a correct expected value.



%



Bounds on type-1 error rates with aberrance in each group are shown in table~\ref{tab:t1r}. Methods B and C necessitate sacrificing bounds on error rates with aberrance in $C_0$ and $C_{0}$,$C_{0}'$ respectively. The bound on error with aberrance in $C_{1}'$ improves through methods A-C. In the methods section, it is shown that the type-1 error with aberrance in $C_{0}'$ decreases from methods A to B, and the error with aberrance in $C_{1}'$ increases from A through C, although the upper bound is the same for both.

\begin{table}[h]
\centering
\caption{Upper bounds on type 1 error rates with aberrance in cohorts, with $\beta>\beta^*>\beta^\perp$}
\label{tab:t1r}
\begin{tabular}{lccccc}
  & \multicolumn{5}{c}{Aberrant}  \\
  & None & $C_0$ &  $C_0'$ & $C_1$ & $C_{1}'$ \\
M. A & $P_{0}$ & $\beta$ & $\alpha$ & $\beta$ & $\alpha$ \\
M. B & $P_{0}$ & 1 & $\alpha$ & $\beta^*$ & $\alpha$ \\
M. C & $P_{0}$ & 1 & 1 & $\beta^\perp$ & $\alpha$   
\end{tabular}
\end{table}

\subsection*{Simulation}


The power difference between methods B and A was analysed systematically across a range of values of $(n_{0},n_{1},n_{0}',n_{1}')$. Average power difference and maximum power difference were compared (see methods section). Figure~\ref{fig:power} shows power difference at various study sizes for typical $\alpha$, $\beta$, $\gamma$ values ($\alpha=5 \times 10^{-6}$, $\beta=5 \times 10^{-4}$, $\gamma=5 \times 10^{-8}$) and minor allele frequency 0.1. The difference is typically highest when the ratio of controls to cases is high in the discovery cohort and low or equal in the replication cohort, and the number of cases in the discovery cohort is larger than the number in the replication cohort. Power to detect SNPs in $H_{1}$ is typically highest in method C, second-highest in method B, and lowest in method A. 


\begin{figure}
\begin{subfigure}[b]{\textwidth}
\begin{subfigure}[b]{0.4\textwidth}
\includegraphics[width=\textwidth]{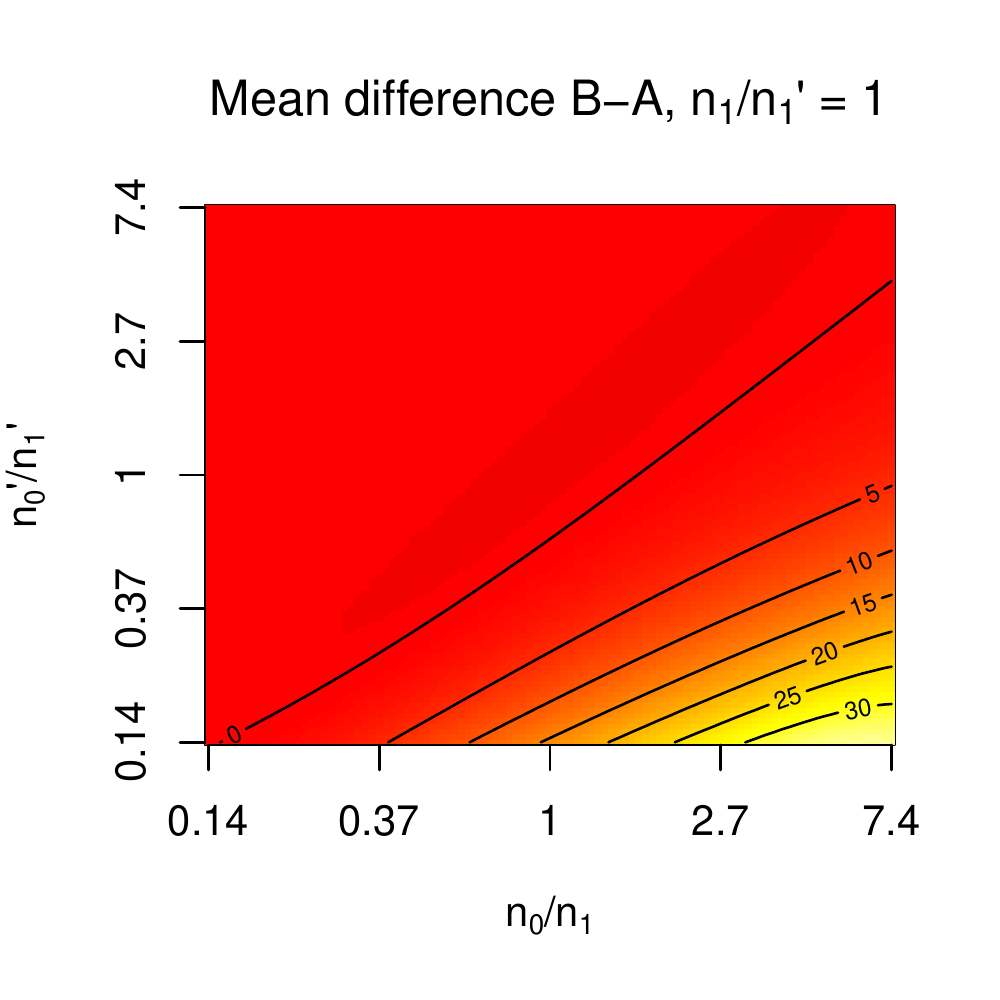}
\end{subfigure}
\begin{subfigure}[b]{0.4\textwidth}
\includegraphics[width=\textwidth]{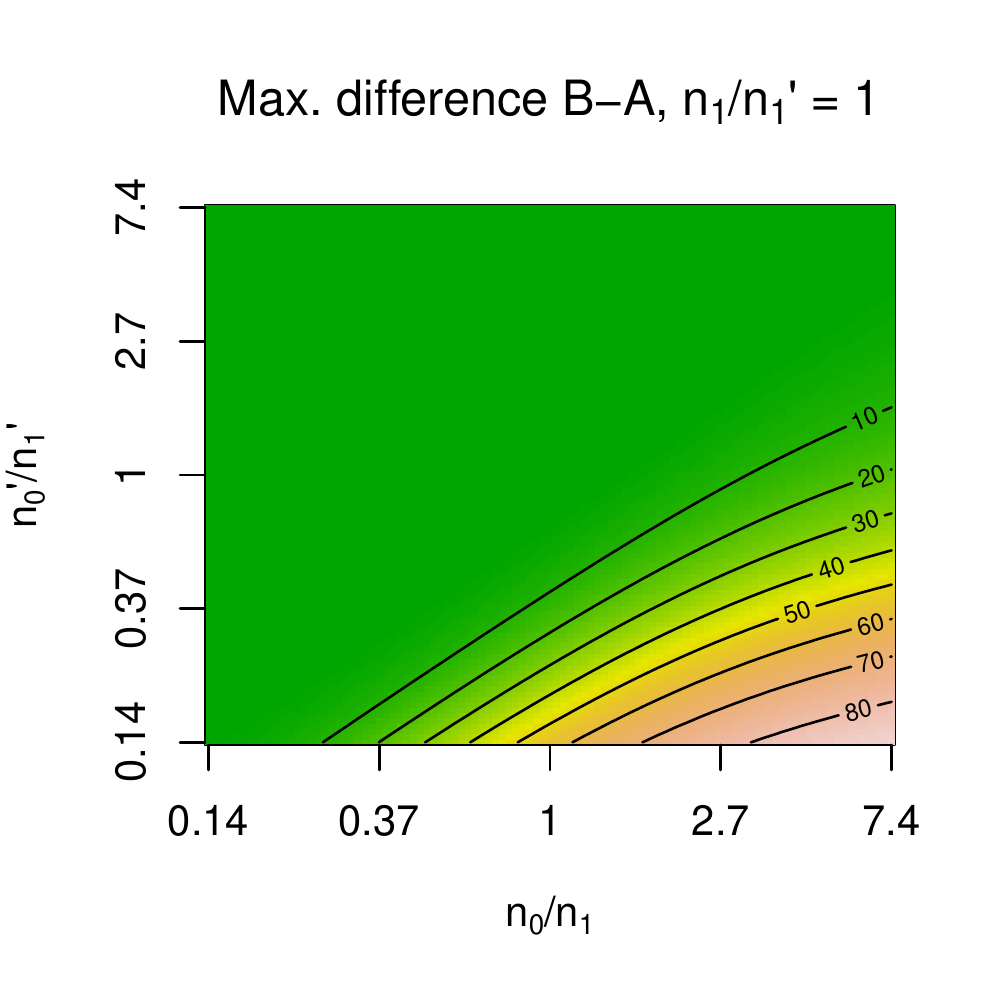}
\end{subfigure}
\begin{subfigure}[b]{0.13\textwidth}
\includegraphics[width=\textwidth]{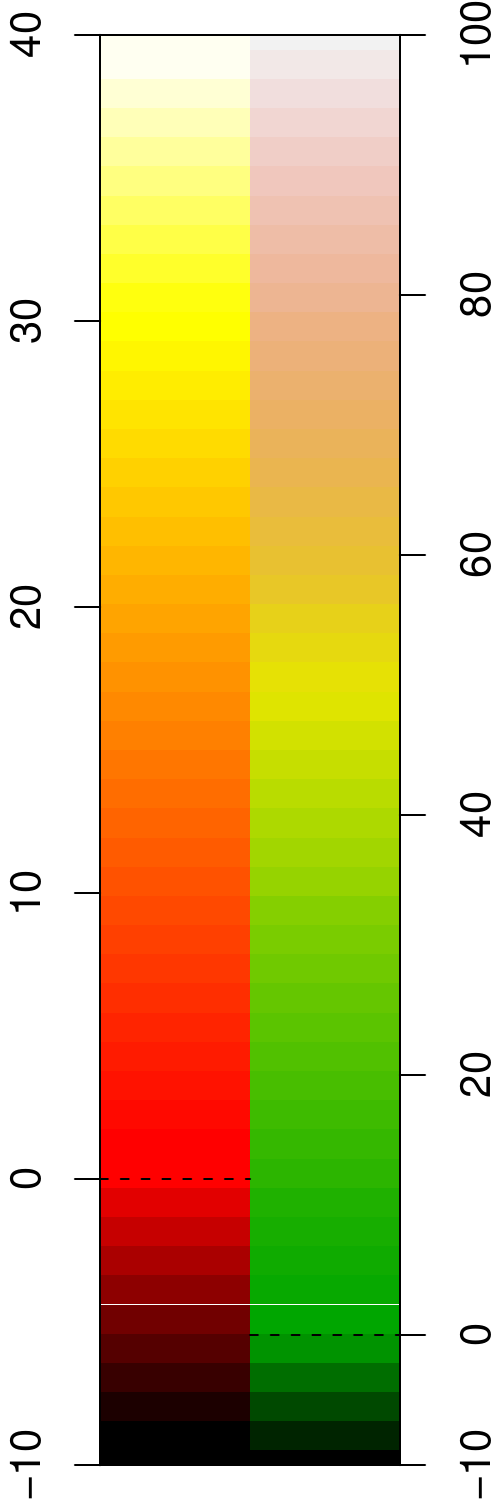}
\end{subfigure}
\end{subfigure}

\begin{subfigure}[b]{\textwidth}
\begin{subfigure}[b]{0.4\textwidth}
\includegraphics[width=\textwidth]{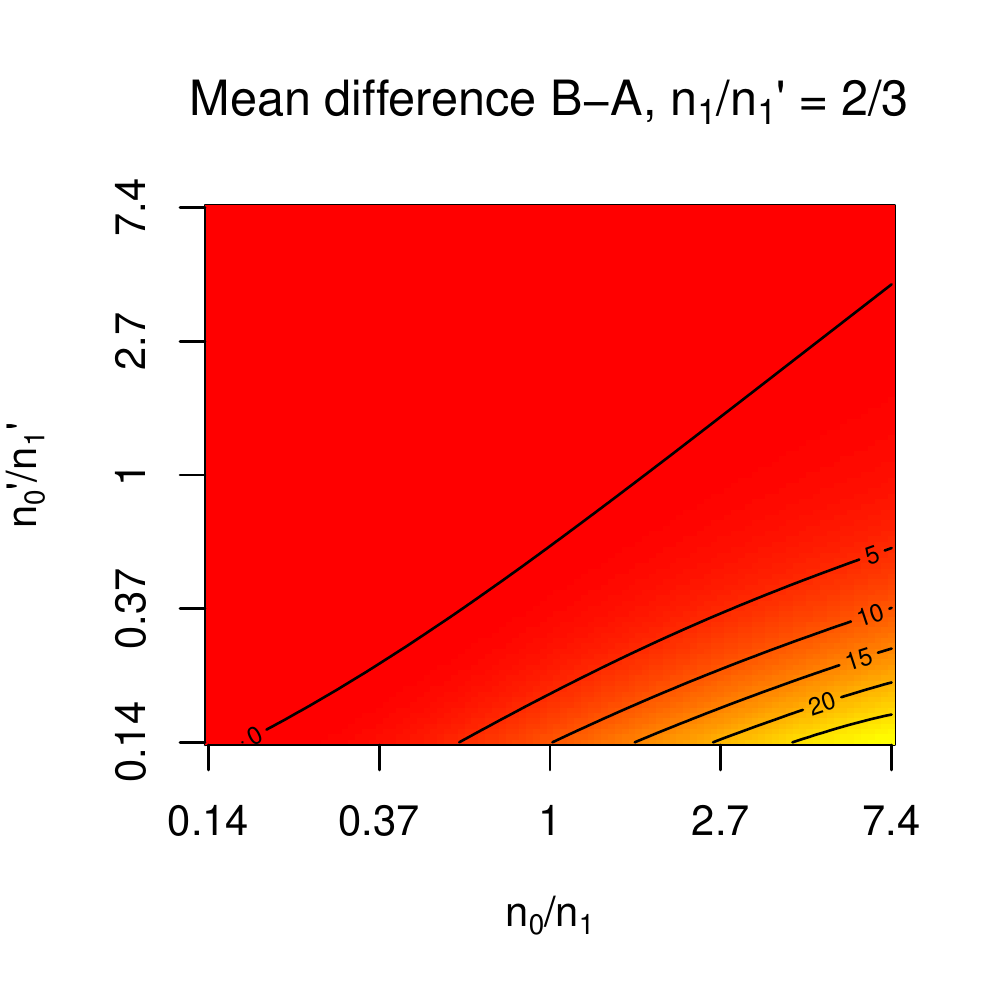}
\end{subfigure}
\begin{subfigure}[b]{0.4\textwidth}
\includegraphics[width=\textwidth]{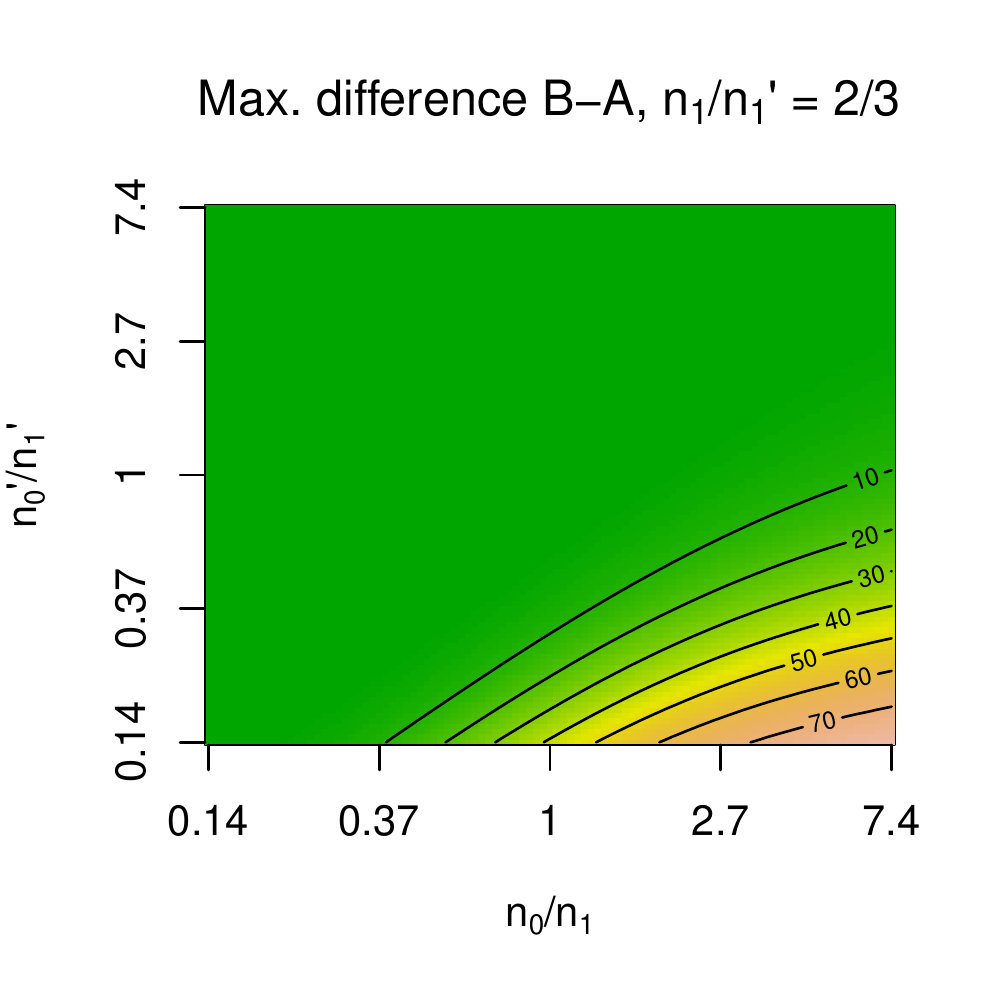}
\end{subfigure}
\end{subfigure}

\begin{subfigure}[b]{\textwidth}
\begin{subfigure}[b]{0.4\textwidth}
\includegraphics[width=\textwidth]{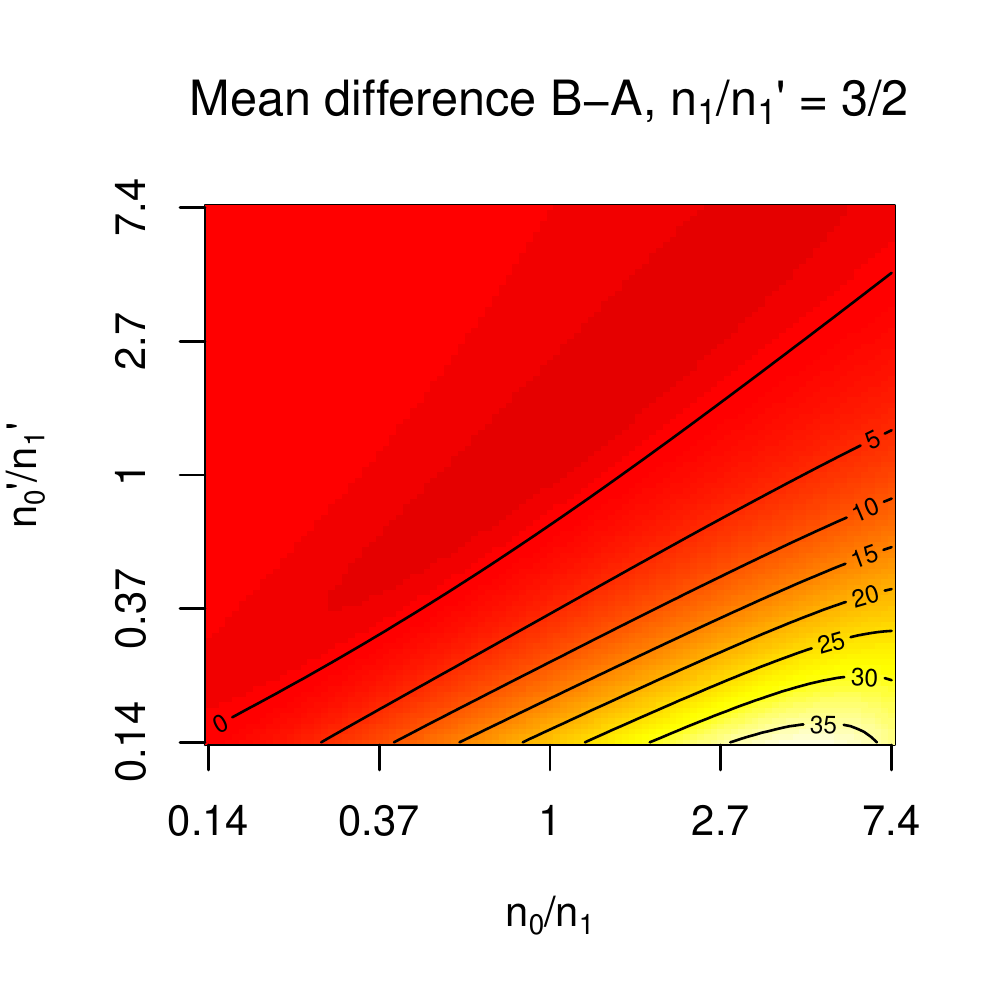}
\end{subfigure}
\begin{subfigure}[b]{0.4\textwidth}
\includegraphics[width=\textwidth]{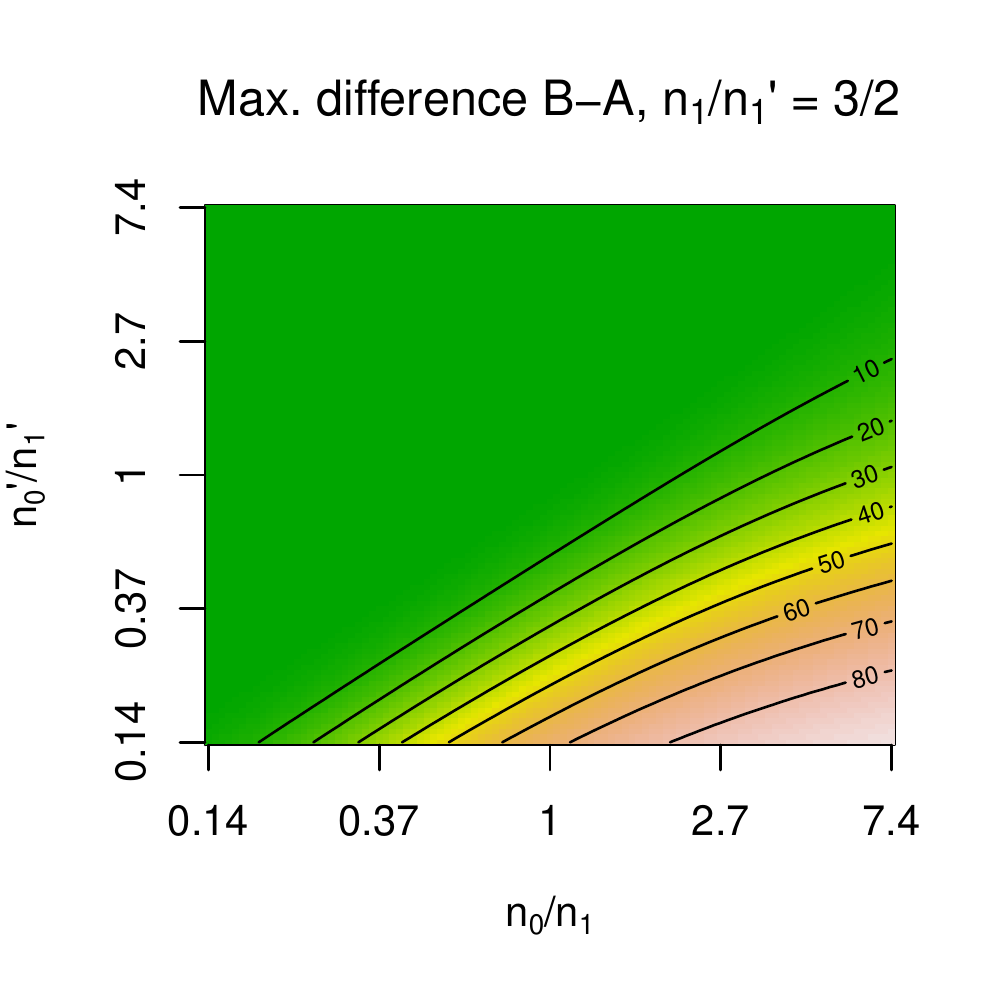}
\end{subfigure}
\end{subfigure}
\caption{General power differences (\%) between methods A and B. Mean power difference is taken as the integral of power difference between methods B and A (see methods section) over $\mathbb{R}$ with respect to log-odds ratio. In all cases, 20 000 samples are used overall for a SNP with MAF 0.1, with cutoffs $\alpha=5 \times 10^{-6}$, $\beta=5 \times 10^{-4}$, $\gamma=5 \times 10^{-8}$. }
\label{fig:power}
\end{figure}

\subsection*{Recommended applications}

To demonstrate areas where this approach is applicable, several examples are constructed or sourced from the GWAS field in which the procedure of sharing controls or cases will improve power or type-1 error profile of the two-stage testing procedure or enable some form of orthogonal replication to be performed.

\subsubsection*{Assumptions}

In order to use method B or C, it must be assumed that cohort $C_{0}$ and $C_{0}'$ are sampled from similar enough populations to be comparable to $C_{1}$ and $C_{1}'$ (possibly with the inclusion of strata or covariates in the genetic risk model). An important caveat of methods B and C is sacrifice of control over errors arising from aberrance in $C_{0}$ (method B) or $C_{0} \cup C_{0}'$ (method C), so an assumption must be made that variables affected by confounding or measurement error in these cohorts are understood to be distinguishable from true associations by quality-control measures only.

Post-hoc assessment of all putative hits should be performed to check for genotyping errors~\cite{anderson10} and assess whether the hit could have arisen from aberrance in $C_{0}$.


\subsubsection*{Conventional GWAS}

Method B is applicable in several cases in large conventional GWAS, particularly when then ratio of controls to cases in the discovery cohort is larger than that in the replication cohort. In a relatively recent GWAS on rheumatoid arthritis~\cite{stahl10} with comparable sample populations for discovery and replication cohorts, method B 
could be used to attain greater power than method A for a fixed type-1 error rate.  
Assuming that summary statistics are well-approximated by binomial tests of allelic differences (so covariates and strata used in computation of summary statistics have only small effects), the improvement in power is around 4\% for SNPs with an odds-ratio of $1.3$, MAF 0.1, and is positive across all odds ratios. More than 2000 additional controls in $C_0'$ would be needed to increase power by this amount (figure~\ref{fig:stahl}). 

Small power advantages such as this may make minimal difference in a single study, although since they require no extra cost, are worth attaining if possible. 
The power of method B is generally considerably higher than method A when $n_{0} > n_{1}$ and $n_{0}' \approx n_{1}'$, 
Power advantages may be more substantial in some cases; for example, a study with $(n_{0},n_{1},n_{0}',n_{1}')=(15 000,5000,5000,5000)$,
method B enables a power increase of up to 8\% (Figure~\ref{fig:demo1}). To achieve comparable performance with method A, around 2000 additional controls would be necessary in the replication cohort. Method B with $(n_{0},n_{0}')=(15000,5000)$ is also more powerful than method A would be if controls were divided equally between $C_{0}$ and $C_{0}'$ (see Figure~\ref{fig:demo1}). 




\subsubsection*{Difficult control ascertainment}

An important application of the method presented in this paper is in studies for which `control' samples are expensive or difficult to ascertain. This is often the case in comparative studies between disease subtypes. In such studies, sharing controls can improve power substantially, especially if a proportion of samples in the replication cohort are falsely assigned to the control cohort (see methods section).


An international GWAS on fronto-temporal dementia in 2014~\cite{ferrari14} is an example in which sharing controls may be beneficial. The study had sample sizes $(n_{0},n_{1},n_{0}',n_{1}')=(4308,2154,5094,1372)$. Control samples in the discovery phase were assessed for current neurological disease, and were used in previous studies on Parkinson's disease, indicating a high degree of reliability. Control samples in the replication phase were collected from the same geographic distribution as cases, but were not explicitly used in previous neurological studies, suggesting better control ascertainment amongst the discovery cohort. 

In this study, sharing controls could allow for a more strongly-ascertained control cohort, and reduce the effects of confounders affecting $C_{1}'$. At typical values $\alpha=1 \times 10^{-4}$, $\beta = 1 \times 10^{-3}$, $\gamma = 5 \times 10^{-8}$, power is nearly equivalent between the two methods assuming all controls are genuine. However, with 10\% misascertainment in $C_{1}'$, the power advantage of method B is up to 5\%.
Given the near-identical distribution of cases in the discovery and validation cohort, cases could alternatively be shared, leading to a power increase of up to 6\%.

\subsection*{Prospective study design}

Studies may be planned and powered with the assumption that samples may be shared. For certain restrictions on sample numbers, this can provide the potential for greater power than would be attainable by restricting to an independent-controls design. For instance, if we seek to validate hits on a GWAS with 10000 controls and 5000 cases, and can afford to genotype a further 10000 samples, power is higher after recruiting 4000 additional controls and 6000 additional cases and sharing controls than can be achieved from any independent-control study design (Figure~\ref{fig:demo3}).

This may be a common scenario if controls are sourced from a known database rather than specifically recruited for the study. 

\begin{figure}[p]
\centering
\begin{subfigure}[b]{0.45\textwidth}
\includegraphics[width=\textwidth] {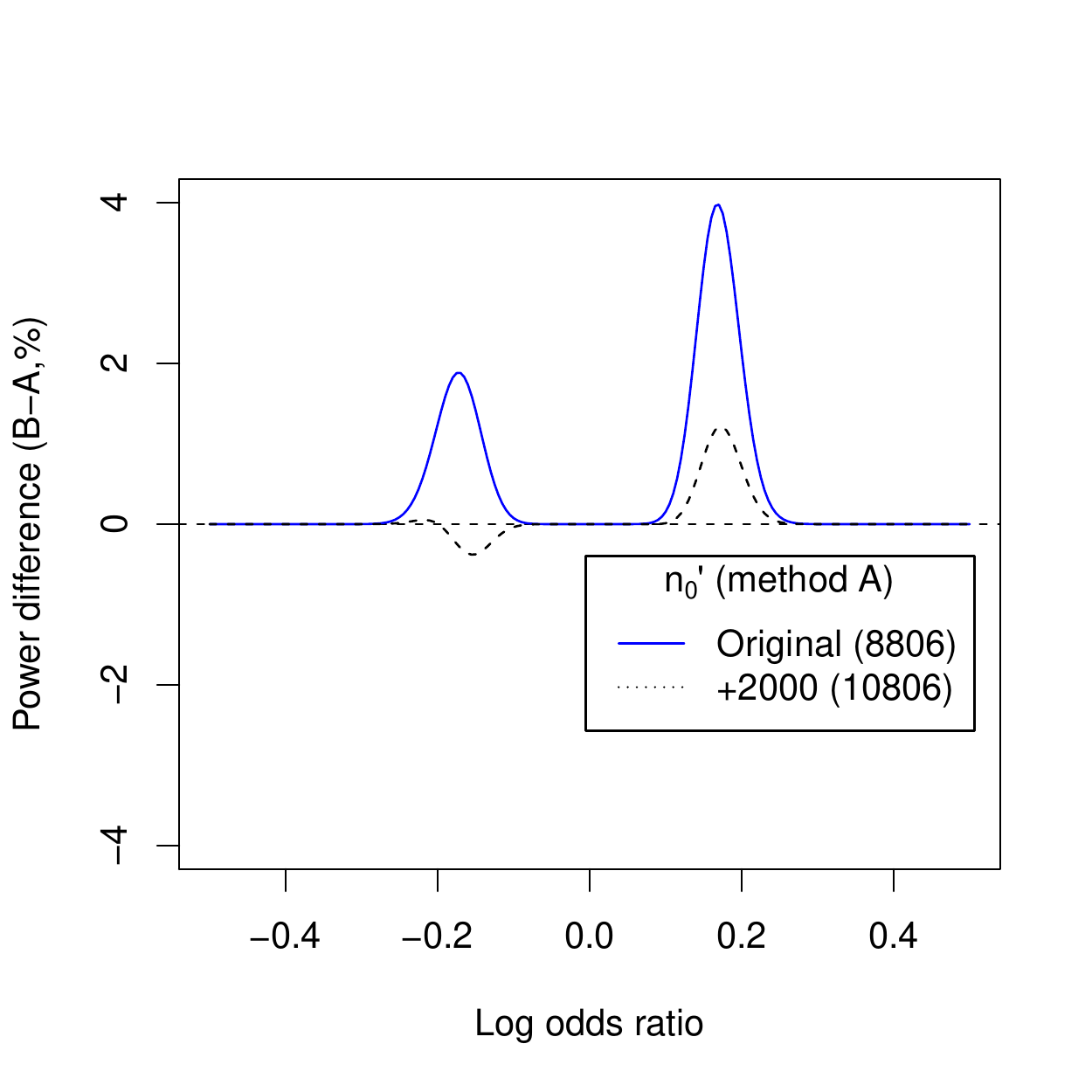}
\caption{Typical GWAS} \label{fig:stahl}
\end{subfigure}
\begin{subfigure}[b]{0.45\textwidth}
\includegraphics[width=\textwidth] {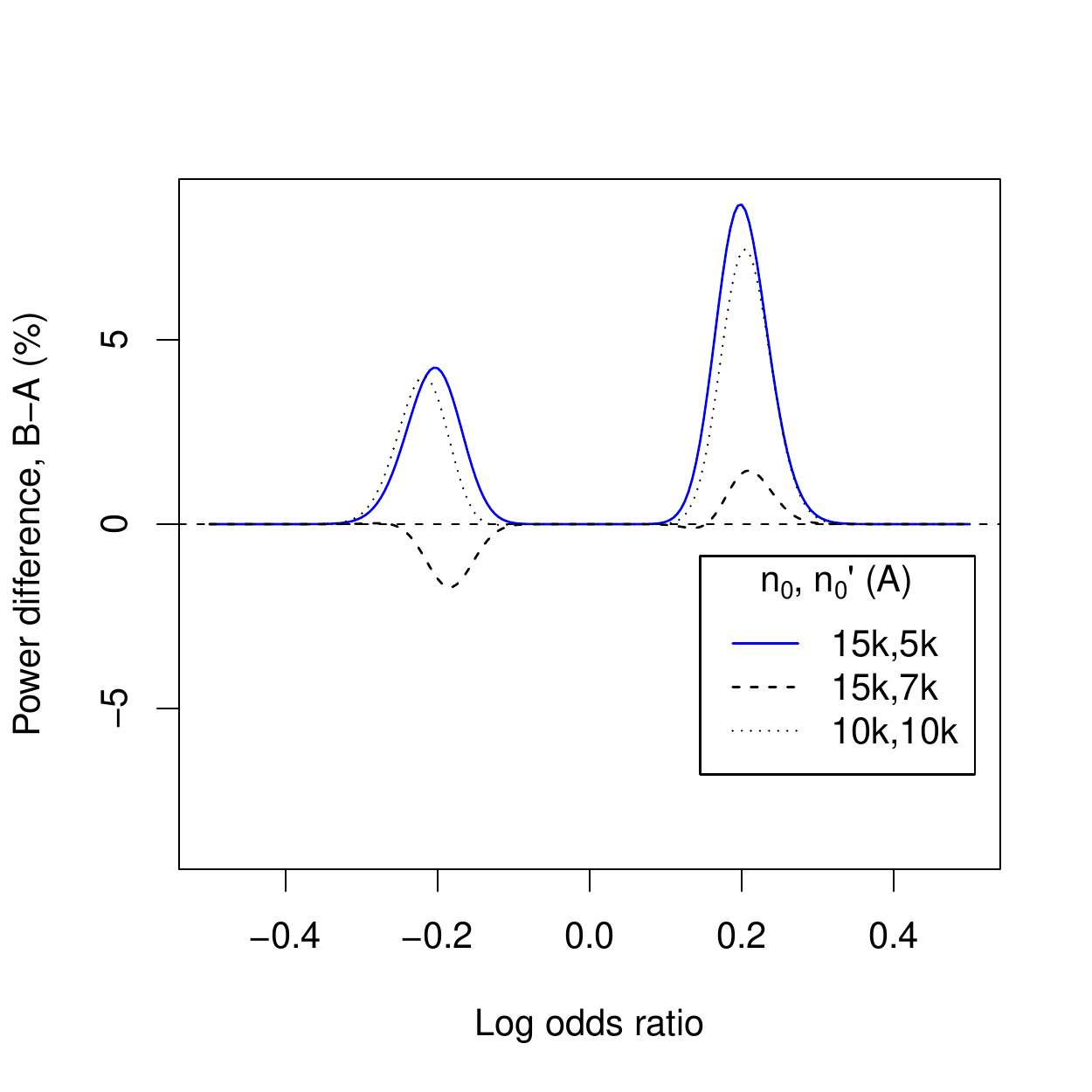}
\caption{Example with $n_{0} \gg n_{0}'$} \label{fig:demo1}
\end{subfigure}
\begin{subfigure}[b]{0.45\textwidth}
\includegraphics[width=\textwidth] {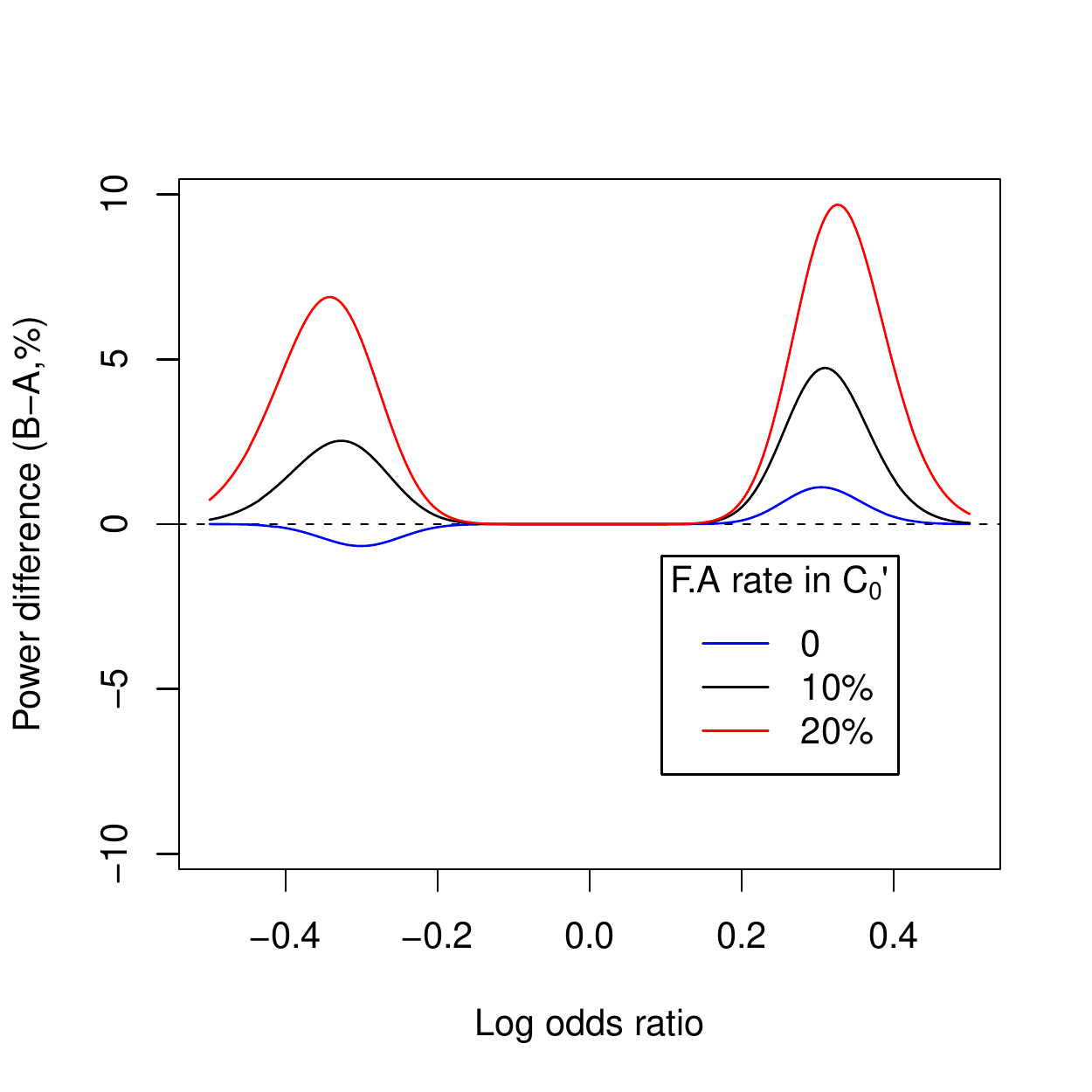}
\caption{Example with incorrect ascertainment} \label{fig:ferrari}
\end{subfigure}
\begin{subfigure}[b]{0.45\textwidth}
\includegraphics[width=\textwidth] {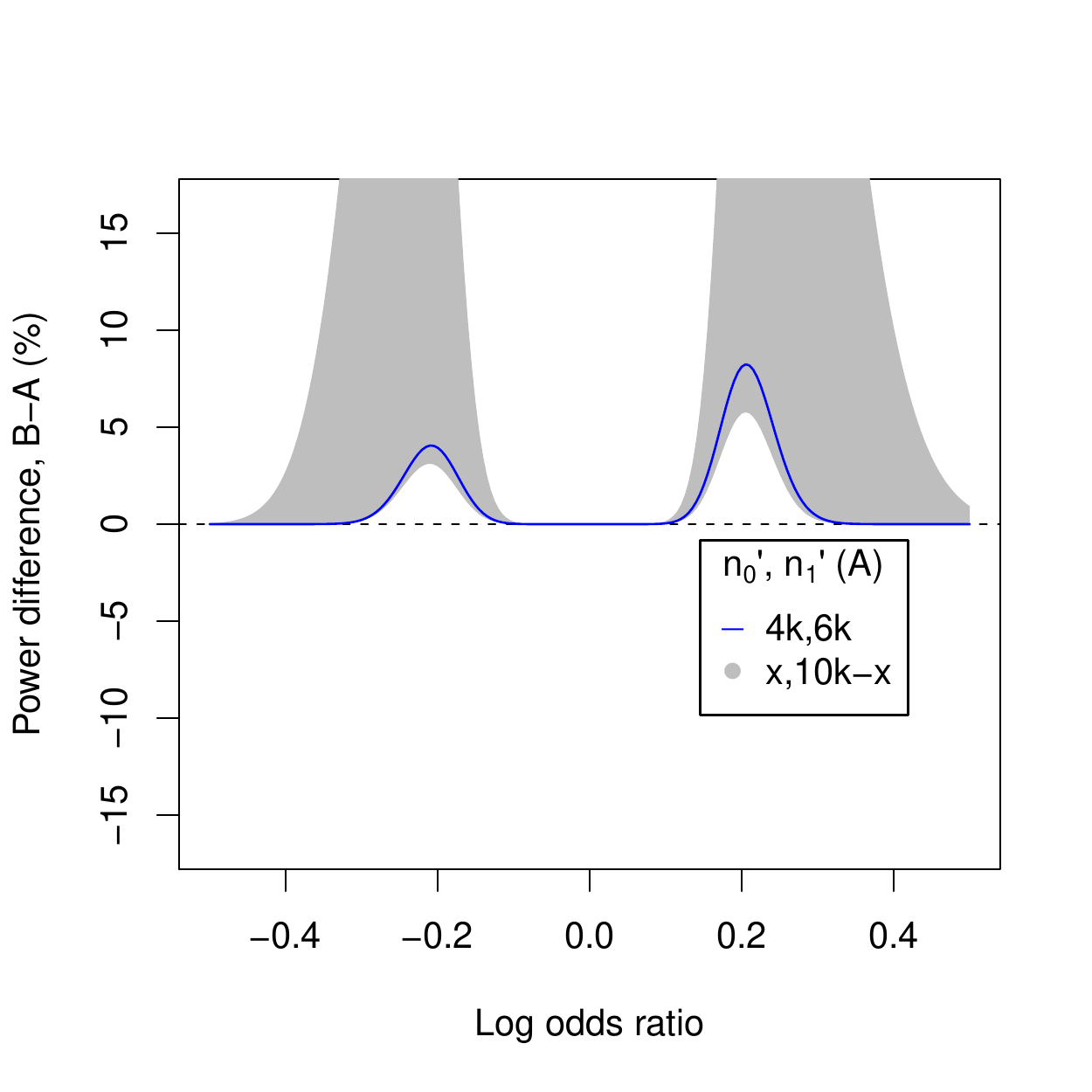}
\caption{Prospective design} \label{fig:demo3}
\end{subfigure}
\caption{
Examples of comparison of power of methods A and B. In all panels, a positive odds ratio corresponds to a deleterious mutation and average MAF is 10\%.
The top two panels show comparisons of method B with $n_0'$ fixed against method A with varying $n_0'$. Panel~\ref{fig:stahl} has $(n_0,n_1,n_0',n_1') = (20169, 5539,8806,6768)$ (values from a GWAS on RA~\cite{stahl10}), and panel~\ref{fig:demo1} $(n_0,n_1,n_0',n_1') = (15000, 5000,5000,5000)$. Both panels use  $(\alpha,\beta,\gamma)=(5 \times 10^{-6},5 \times 10^{-4}, 5 \times 10^{-8})$.
Panel \ref{fig:ferrari} demonstrates the effect of false-ascertainment (F.A) in $C_0'$; when cases are mis-ascertained as controls. In this case, $(\alpha,\beta,\gamma)=(1 \times 10^{-4},1 \times 10^{-3}, 5 \times 10^{-8})$, reflecting values used in the paper~\cite{ferrari14}.
Panel~\ref{fig:demo3} demonstrates a prospective scenario with 10000 samples for replication. Method B with $(n_{0},n_{1})$ as above, $(n_{0}',n_{1}')=(4000,6000)$ is more powerful than any design using method A (grey region; $n_{0}' \in (1000,9000)$; $n_{1}'=10000-n_{0}'$).
}
\label{fig:demos}
\end{figure}

\subsubsection*{Partial replication}



In circumstances where case recruitment is difficult, as in studies of rare diseases, an assessment of replicability may be made by re-testing results from a discovery phase with a new control set only. This can enable the use of control cohorts which only partially match the case cohort.

In a GWAS on pemphigus vulgaris~\cite{sarig12}, a rare disease primarily affecting individuals of Ashkenazi Jewish ethnicity, the discovery cohorts were sampled from Jewish populations, with age- and population- matched controls. Control cohorts were small ($(n_{0},n_{1},n_{0}',n_{1}')=(100,400,59,285)$), potentially due to difficulty recruiting both ethnically- and geographically-matched controls. 

Method C could be used in this instance to enable a larger control set and greater power. If a control cohort of Ashkenazi individuals could be assembled without requiring geographic matching with the case set, it would be inappropriate to use as a sole control cohort against the existing case cohort, due to the potential for geographic confounding. However, such a cohort could be used as either $C_0$ or $C_0'$ in method C, with the existing ethnically- and geographically- matched controls serving as the other cohort. In this way, the power advantage of the larger cohort could be used while maintaining control over potential aberrance in the larger control group.

Method C enables computation of power and type-1 error rates, and comparison to alternative designs with cases split into smaller independent discovery and validation cohorts (method A). Testing a case cohort against two separate control cohorts is almost always more powerful for a fixed type-1 error rate than splitting the case cohort in two and performing method A (see supplementary figures~\ref{fig:powerAC},\ref{fig:powerBC}).




\section*{Discussion}

This paper proposes a method to improve efficiency of data use in a replication procedure, adding to the body of methods for comparison of high-dimensional case-control studies. For many common study sizes, the method can reduce the cost of replication, or increase power of discovery. The adapted method is simple to apply, only requiring modification of a single association threshold. A standard replication procedure (or more general comparison of case-control studies) with independent control datasets does not make use of the information that expected values of variables in control datasets are, in principle, the same. In this way, the same dataset can in theory yield more information when controls are shared.

The most important caveat of these methods is the loss of systematic type-1 error rate control for null SNPs which are aberrant in $C_{0}$. Control of such errors must not be sacrificed entirely, but in some circumstances it may be satisfactory to assess such errors on a SNP-by-SNP basis. Such assessment is important and standard for all proposed GWAS hits under any method~\cite{wtccc07} in the interests of quality control. In method C, control over aberrance in $C_0'$ is additionally lost; however, since this method is largely applicable when $C_0 \cup C_0'$ is a single homogenous control (or case) cohort, there is no way that aberrance in the cohort can be systematically identified by comparison with other cohorts.

Somewhat better control of the type-1 error rate can often be achieved for SNPs with aberrance in $C_{1}$ or $C_{0}'$. This may incentivise the use of this method when confidence in the representativeness of these cohorts is low compared to that of $C_{0}$. The type 1 error rate is somewhat increased for SNPs with aberrance in $C_{1}'$, although as it remains bounded by $\alpha$, this increase is not a major problem.

The two-stage validation procedure is similar to a meta-analysis of the discovery and validation experiments, for which several adaptations to shared-control designs have been proposed~\cite{lin09,han16}. However, there are several important distinctions which necessitate an alternative approach in this case. Firstly, not all variables are measured in the second (replication) study; we are restricted to analysis of variables reaching a given observed effect size. Secondly, the studies to be `meta-analysed' are not complete, in the sense that there may be residual confounding; a strong effect size in the meta-analysis alone is not adequate evidence for association and some level of association (with consistent direction) is additionally required in both constituent studies. 

The method is inapplicable when replication is performed on cohorts from completely distinct geographic groups, although there can be some difference in geographic distribution between control sets if this is controlled for in computing summary statistics. The method is most applicable when control groups are sampled from similar populations and genotyped on similar platforms.

The widespread discoveries of the GWAS field have led to corresponding increases in complexity of phenotypic definitions, with ever-finer delineations of disease types of ever-rarer prevalence. The genetic analysis of such complex phenotypes is necessarily comparative; there is little use understanding the genetics of a rare disease subtype except in the context of the genetics of the disease in general. Such analyses necessitate GWAS and other comparative studies between rare phenotypic types~\cite{liley16}, with `controls' meaning the better-characterised disease subphenotype in this sense, as well as between cases and controls. Rare disease subtypes are often afflicted with ascertainment difficulties, leading to varying degrees of expected aberrance in disease cohorts. Within this paradigm, the applicability of this method is likely to expand.

\section*{Acknowledgments}

I would like to thank Dr Chris Wallace and Dr Jenn Asimit for helpful comments and review of this work.
JL is funded by the NIHR Cambridge Biomedical Research Centre and is on the Wellcome Trust PhD programme in Mathematical Genomics and Medicine at the University of Cambridge. He is also supported by the Wallace group, funded by the Wellcome Trust (089989,107881, CW) and the MRC (MC\_UP\_1302/5, CW). 
The funders had no role in study design, data collection and analysis, decision to publish, or preparation of the manuscript. 

\section*{Conflicts of interest}

None declared

\clearpage

\section*{Methods}

\subsection*{Definitions}

Denote $z_{x}$ for $x \in \{d,r,s,m,c \} $ as the signed z-score ($\pm \Phi^{-1}(p_{x}/2)$) corresponding to $p_{x}$, and $z_{x}$ for $x \in {\alpha,\beta,\beta^*,\gamma}$ as the positive corresponding threshold $-\Phi^{-1}(x/2)$, where $\Phi, \Phi^{-1}$ are the standard normal CDF and quantile functions. 
%
%
Other than $(z_{d},z_{r})$, all pairs of z-scores are correlated under $H_0^\cap$, with correlation estimable from sample sizes or empirically if covariates are used (Appendix~\ref{apx:zcor}). Denote $\rho_{xy}$ as the correlation between $z_{x}$ and $z_{y}$, $(x,y) \in \{d,r,s,m\}^{2}$, and set
%
%
\begin{equation}
\Sigma_{A} = var \left(  ( z_{d} \, z_{r} \, z_{m} )^t \right) \hspace{40pt}
\Sigma_{B} = var \left( ( z_{d} \, z_{s} \, z_{m} )^t \right) \hspace{40pt}
\Sigma_{C} = var \left( ( z_{c} \, z_{s} \, z_{m} )^t \right) \label{eq:sigcdef}
\end{equation}
For $i \in \{ d,r,s,m,c \}$ define $\zeta_{i}=E(z_{i})$, where the expectation is conditional on the SNP in question. For SNPs in $H_{0}^{\cap}$, $\zeta_{i} \equiv 0$, but this may not hold for SNPs in $H_{0}^{\cup} \setminus H_{0}^{\cap}$. In theoretical working, aberrance in groups is characterised by values $\zeta_{i}$ rather than log-odds ratios, noting that the values $\zeta_{i}$ are asymptotically proportional to the corresponding log-odds ratios. Define $R_{A}$, $R_{B}$, $R_C$ as the false-positive rates for a SNP of interest in methods A, B and C respectively. 

\subsection*{General type 1 error rate}

The values $\beta^*$, $\beta^\perp$ are chosen to satisfy 
\begin{align}
2\int_{z_{\alpha}}^{\infty} \int_{z_{\beta^*}}^{\infty} \int_{z_{\gamma}}^{\infty} N_{\Sigma_{B}}  \left( \begin{smallmatrix} z_{d} \\ z_{r} \\ z_{m} \end{smallmatrix} \right) dz_{m} dz_{r} dz_{d} &= 2\int_{z_{\alpha}}^{\infty} \int_{z_{\beta^\perp}}^{\infty} \int_{z_{\gamma}}^{\infty} N_{\Sigma_{C}}  \left( \begin{smallmatrix} z_{d} \\ z_{r} \\ z_{m} \end{smallmatrix} \right) dz_{m} dz_{r} dz_{d} \nonumber \\ 
&= 2\int_{z_{\alpha}}^{\infty} \int_{z_{\beta}}^{\infty} \int_{z_{\gamma}}^{\infty} N_{\Sigma_{A}}  \left( \begin{smallmatrix} z_{d} \\ z_{r} \\ z_{m} \end{smallmatrix} \right) dz_{m} dz_{r} dz_{d} \nonumber \\
&= Pr(p_{d}<\alpha,p_{r}<\beta,p_{m}<\gamma|H_{0}^{\cap}) \label{eq:betastardef} 
\end{align}
thus conserving the type 1 error rate (denoted $P_0$) against $H_{0}^{\cap}$ between methods (Figure~\ref{fig:setup}). If no threshold is used on $p_m$ (ie, $\gamma=1$), then $\beta^*$, $\beta^\perp$ satisfy
\begin{equation}
Pr(p_{d}<\alpha,p_{s}<\beta^*|H_{0}^{\cap}) = Pr(p_{c}<\alpha,p_{s}<\beta^\perp|H_{0}^{\cap}) =Pr(p_{d}<\alpha,p_{r}<\beta|H_{0}^{\cap}) = \alpha \beta \label{eq:betastardef2}
\end{equation}
since $z_{d} \ci z_{r}|H_{0}^{\cap}$. Definition~\ref{eq:betastardef} will be considered a generalisation of definition~\ref{eq:betastardef2}, with results established first for $\beta^*$ as per definition~\ref{eq:betastardef2} and extending where possible to definition~\ref{eq:betastardef}.

For $\beta^*$ defined as per definition~\ref{eq:betastardef2} we have (see Appendix~\ref{apx:infbetastar})
\begin{equation}
\lim_{z_{\alpha} \to \infty} \frac{z_{\beta^{*}}}{\sqrt{1-\rho_{ds}^{2}}z_{\beta} + \rho_{ds} z_{\alpha}} = 1  \hspace{40pt} \lim_{z_{\alpha} \to \infty} \frac{z_{\beta^\perp}}{\sqrt{1-\rho_{cs}^{2}}z_{\beta} + \rho_{cs} z_{\alpha}} = 1
\end{equation}
approaching from above, so $z_{\beta^{*}} > max \left( z_{\beta},\sqrt{1-\rho_{ds}^{2}}z_{\beta} + \rho_{ds} z_{\alpha} \right)$, $z_{\beta^{\perp}} > max \left( z_{\beta},\sqrt{1-\rho_{cs}^{2}}z_{\beta} + \rho_{cs} z_{\alpha} \right)$. As defined by equation~\ref{eq:betastardef2}, $z_{\beta^*}$, $z_{\beta^\perp}$ are also  asymptotically linear in $z_{\alpha}$, $z_{\gamma}$, $z_{\beta}$ as the former two tend to $\infty$, with some constraints (Appendix~\ref{apx:infbetastar}), although the limit does not necessarily approach from above. For both definitions, $\beta^\perp < \beta^*<\beta$ (Appendix~\ref{apx:betabound})

\begin{figure}[p]
\begin{center}
\includegraphics[width=0.5\textwidth] {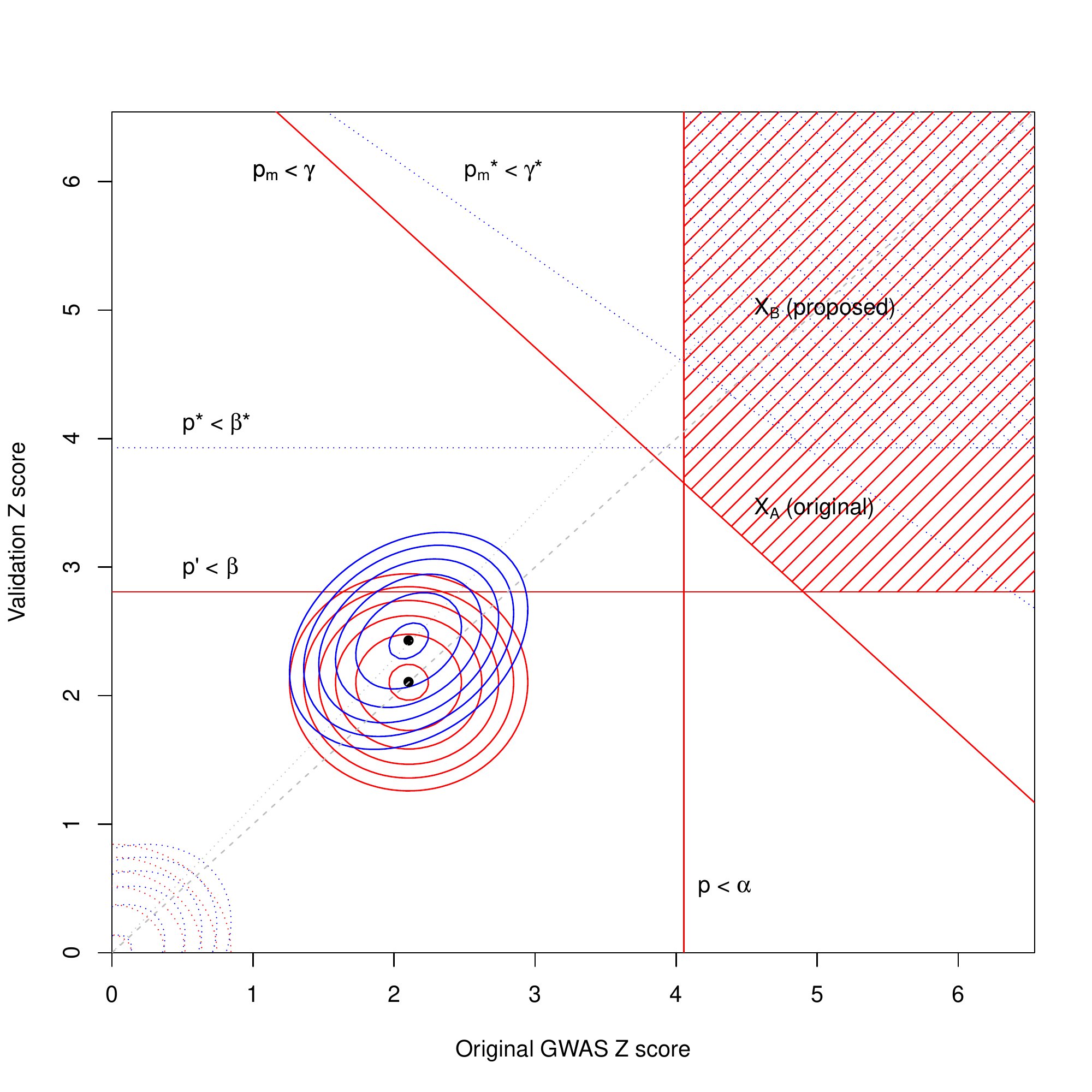}
\end{center}
\caption{
Replication with shared controls. Red and blue shaded areas are regions where a pair of observed $Z$ scores are deemed a `hit' in the $(+,+)$ quadrant under method A/B  respectively. The value $z_{m}$ is almost linearly dependent on $(z_{d}, z_{r})$ and on $(z_{d},z_{s})$ (Appendix~\ref{apx:zcor}). Solid red/blue ellipses indicate contours of the distribution of observed $Z$ scores for a typical non-null SNP under methods A and B, and dashed ellipses indicate contours for a null SNP. 
}
 \label{fig:setup}
\end{figure}

\subsection*{Study sizes, odds ratios and allele frequencies}

Consider a study with $n_{0}$ controls and $n_{1}$ cases, with underlying allele frequencies $\mu_{0}$ and $\mu_{1}$ in cases and observed allele frequencies $m_{0}$, $m_{1}$. Let $Z$ be a signed Z-score derived from a GWAS p-value against the null hypothesis $\mu_{0}=\mu_{1}$. To first order, 
\begin{equation}
E(Z) = \sqrt{\frac{2 n_{0} n_{1}}{n_{0}+n_{1}}} \frac{\mu_{1}-\mu_{0}}{\sqrt{\bar{\mu}(1-\bar{\mu})}}
\end{equation}
where $\bar{\mu} = \frac{n_{0}\mu_{0}+n_{1} \mu_{1}}{n_{0}+n_{1}}$. Hence
%
%
\begin{eqnarray}
&\zeta_{d} = \sqrt{\frac{2 n_{0}n_{1}}{n_{0}+n_{1}}}\frac{\mu_{1}-\mu_{0}}{\sqrt{\bar{\mu}(1-\bar{\mu})}} \hspace{20pt}
\zeta_{r} = \sqrt{\frac{2 n_{0}'n_{1}'}{n_{0}'+n_{1}'}}\frac{\mu_{1}'-\mu_{0}'}{\sqrt{\bar{\mu}(1-\bar{\mu})}} \nonumber \\
&\zeta_{s} = \sqrt{\frac{2 (n_{0}+n_{0}')n_{1}'}{n_{0}+n_{0}'+n_{1}'}}\frac{\mu_{1}'-\frac{\mu_{0}n_{0}+\mu_{0}'n_{0}'}{n_{0}+n_{0}'}}{\sqrt{\bar{\mu}(1-\bar{\mu})}} \hspace{20pt}
\zeta_{c} = \sqrt{\frac{2 (n_{0}+n_0')n_{1}}{n_{0}+n_0'+n_{1}}}\frac{\mu_{1}-\frac{\mu_{0}n_{0}+\mu_{0}'n_{0}'}{n_{0}+n_{0}'}} {\sqrt{\bar{\mu}(1-\bar{\mu})}} \nonumber \\
&\zeta_{m} = \sqrt{\frac{2 (n_{0}+n_{0}')(n_{1}+n_{1}')}{n_{0}+n_{0}'+n_{1}+n_{1}'}}\frac{\frac{\mu_{1}n_{1}+\mu_{1}'n_{1}'}{n_{1}+n_{1}'}-\frac{\mu_{0}n_{0}+\mu_{0}'n_{0}'}{n_{0}+n_{0}'}}{\sqrt{\bar{\mu}(1-\bar{\mu})}} \label{eq:zetadef}
\end{eqnarray}
where $\bar{\mu}$ varies between definitions (though it is sometimes taken to be approximately equal). These formulae allow $\zeta_{i}$ to be estimated in empirical computations. Estimation of $\zeta_i$ is more complex if covariates or strata are used in the computation of $z_i$ (appendix~\ref{apx:zcor}).

In all empirical computations, systematic allelic differences (whether due to aberrance or true association) are characterised by odds ratios and minor allele frequency rather than $\zeta_{i}$. Values $\mu_{0}$ and $\mu_{1}$ (in general terms) are readily computable from odds-ratios  $R=\frac{\mu_{1}(1-\mu_{0})}{\mu_{0}(1-\mu_{1})}$ and minor allele frequency $\mu=\frac{\mu_{0}+\mu_{1}}{2}$. The use of this characterisation of minor allele frequency (rather than the `weighted' characterisation $\frac{n_{0}\mu_{0}+n_{1}\mu_{1}}{n_{0}+n_{1}}$) is so that the correspondence between $(R,\mu)$ and $(\mu_{0},\mu_{1})$ is independent of $n_{0},n_{1}$; for instance, for a disease-associated variant satisfying $\mu_{1}=\mu_{1}'$, $\mu_{0}=\mu_{0}'$ will have identical odds ratios $\frac{\mu_{1}(1-\mu_{0})}{\mu_{0}(1-\mu_{1})}$ and $\frac{\mu_{1}'(1-\mu_{0}')}{\mu_{0}'(1-\mu_{1}')}$ between the discovery and validation experiments, but generally different weighted MAFs $\frac{n_{0}\mu_{0}+n_{1}\mu_{1}}{n_{0}+n_{1}}$ and $\frac{n_{0}'\mu_{0}'+n_{1}'\mu_{1}'}{n_{0}'+n_{1}'}$.

\subsubsection*{False ascertainment}

In general, for a true association, $\mu_{0}=\mu_{0}'$ and $\mu_{1}=\mu_{1}'$. If some proportion $\kappa$ of samples in $C_{0}'$ are incorrectly assigned and come from the case population, then $\mu_{0}' = (1-\kappa)\mu_{0} + \kappa \mu_{1}$. This lowers the absolute values of $\zeta_{r}$, $\zeta_{s}$ and $\zeta_{m}$, reducing the power to detect the SNP. 

\subsection*{Empirical computations}
\label{sec:empirical}

Define $N_{\Sigma}(\mathbf{z})$ as the $pdf$ of the multivariate normal with mean $0$ and variance $\Sigma$ at $\mathbf{z}$. Determination of covariance is described in Appendix~\ref{apx:zcor}.
Given $\zeta_{d}$, $\zeta_{r}$, $\zeta_{s}$, $\zeta_{m}$, the probability of rejecting the null for a given SNP using method A is %
\begin{align}
&\int_{z_{\alpha}-\zeta_{d}}^{\infty} \int_{z_{\beta} - \zeta_{r}}^{\infty} \int_{z_{\gamma} - \zeta_{m}}^{\infty} N_{\Sigma_{A}} \left( (z_{d} \, z_{r} \, z_{m} )^t\right) dz_{m} dz_{r} dz_{d} \nonumber \\ 
&+ \int_{z_{\alpha}+\zeta_{d}}^{\infty} \int_{z_{\beta} + \zeta_{r}}^{\infty} \int_{z_{\gamma} + \zeta_{m}}^{\infty} N_{\Sigma_{A}} \left( (z_{d} \, z_{r} \, z_{m} )^t\right) dz_{m} dz_{r} dz_{d}
\end{align}
and using method B
\begin{align}
&\int_{z_{\alpha}-\zeta_{d}}^{\infty} \int_{z_{\beta} - \zeta_{s}}^{\infty} \int_{z_{\gamma} - \zeta_{m}}^{\infty} N_{\Sigma_{B}} \left( (z_{d} \, z_{s} \, z_{m} )^t\right) dz_{m} dz_{s} dz_{d} \nonumber \\ 
&+ \int_{z_{\alpha}+\zeta_{d}}^{\infty} \int_{z_{\beta} + \zeta_{s}}^{\infty} \int_{z_{\gamma} + \zeta_{m}}^{\infty} N_{\Sigma_{B}} \left( (z_{d} \, z_{s} \, z_{m} )^t\right) dz_{m} dz_{s} dz_{d}
\end{align}
If $\frac{n_{0}}{n_{1}}=\frac{n_{0}'}{n_{1}'}$, matrix $\Sigma_{A}$ is singular (Appendix~\ref{apx:zcor}), in which case $z_{m}=\rho_{dm} z_{d} + \rho_{vm}z_{v}$ and the expression above may be reduced to a two-dimensional integral over a more complex region (Figure~\ref{fig:setup}). Matrix $\Sigma_{C}$ is generally singular, so the formula $z_{m} = \frac{\rho_{cs}\rho_{sm}-\rho_{cm}}{\rho_{cs}^2 - 1} z_{d} + \frac{\rho_{cs}\rho_{cm} - \rho_{sm}}{\rho_{cs}^{2} - 1} z_{s}$ 
is used to reduce the integral in a similar way. A similar formula may be used if $\Sigma_B$ is nearly singular.


In Figure~\ref{fig:power}, mean power difference is determined as the integral of the power difference with respect to the log-odds ratio over the real line.  

\subsection*{Type 1 error rates}
\subsubsection*{Aberrance in $C_{1}$}

For SNPs aberrant in only $C_{1}$ we have $\zeta_{d} \neq 0$, $\zeta_{c} \neq 0$, $\zeta_m \neq 0$, and $\zeta_{r}=\zeta_{s}=0$.

$R_{A}$, $R_{B}$, $R_{C}$ can be considered as functions of $\zeta_{d}$. As $\zeta_{d} \to 0$, $R_{A},R_{B},R_{C} \to P_{0}$ (equation~\ref{eq:betastardef}). As $\zeta_{d} \to \pm \infty$, $R_{A} \to \frac{\beta}{2}$, $R_{B} = \frac{\beta^{*}}{2}$ and $R_{C} = \frac{\beta^\perp}{2}$. For positive $\zeta_{d}$ both $R_{A}$ and $R_{B}$ are increasing (and both are symmetric in $\zeta_{d}$) so $R_{A}<\frac{\beta}{2}$,  $R_{B} < \frac{\beta^{*}}{2}$, $R_{C} < \frac{\beta^{\perp}}{2}$ for all $\zeta_{d}$.

Since $\beta^\perp < \beta^* < \beta$ (often substantially), methods B and C are generally better at rejecting $H_{0}^{\cap}$ for such SNPs. In the simplified case where $z_{\gamma}=1$, $R_{A} \geq R_{B}$ universally (Appendix~\ref{apx:abc1}. This typically holds for all $z_{\gamma}$, except for small deviations in pathological cases.

In general, we consider aberrance which is only still present after any strata or covariates have been accounted for in the computation of $z$ scores. If strata or covariates remove the effective aberrance between groups, the type-1 error rate is equivalent to that under $H_{0}^{\cap}$.

\subsubsection*{Aberrance in $C_{1}'$}

For SNPs aberrant in $C_{1}'$, we have $\zeta_{d}=0$, $\zeta_{c}=0$, $\zeta_{r} \neq 0$, $\zeta_{s} \neq 0$ and $\zeta_{m} \neq 0$. 

Again, $R_{A},R_{B},R_{C} \to P_{0}$ as $\zeta_{r} \to 0$. As $\zeta_{r} \to \pm \infty$, $R_{A},R_{B},R_{C} \to \frac{\alpha}{2}$, and both are bounded by $\frac{\alpha}{2}$. Although $R_{B}$ and $R_{C}$ are typically higher than $R_{A}$ in this case, since both have the same (typically conservative) upper bound, this is not typically a large sacrifice in type 1 error. 

In the simplified case where $\gamma =1$, an approximate upper bound on $R_{B}-R_{A}$ is given by (Appendix~\ref{apx:c1vaberrance})
\begin{equation}
\frac{\alpha}{2 \sqrt{2 \pi}} \left(\frac{k}{\sqrt{1-\rho^{2}}} - 1 \right) z_{\beta} \ll \frac{\alpha}{2} \label{eq:c1v_aberrance}
\end{equation}
where
\begin{equation}
 k =\frac{\zeta_{s}}{\zeta_{r}} \approx \sqrt{\frac{(n_{0}+n_{0}')(n_{0}'+n_{1}')}{n_0'(n_0+n_0'+n_1')}}
\end{equation}
In practice, there is typically a similarly small difference between $R_{C}$, $R_{B}$ and $R_{A}$ in the general case.


\subsubsection*{Aberrance in $C_{0}'$}

For SNPs aberrant in $C_{0}'$, $\zeta_{d}=0$, $\zeta_{r} \neq 0$, $\zeta_{c} \neq 0$, $\zeta_{s} \neq 0$ and $\zeta_{m} \neq 0$. As for SNPs with aberrance in $C_{1}'$, $R_{A},R_{B},R_{C} \to P_{0}$ as $\zeta_{r} \to 0$ and as $\zeta_{r} \to \pm \infty$, $R_{A},R_{B} \to \frac{\alpha}{2}$, both bounded above by $\frac{\alpha}{2}$. $R_{C}$, however, tends to 1 as $\zeta_{d} \to \infty$. 

In method B the cohort $C_{0}$ has a correcting effect on the replication study, meaning $|\zeta_{s}|<|\zeta_{r}|$ and $R_{B}<R_{A}$. 

For the simplified case where $\gamma=1$, a similar bound to ~\ref{eq:c1v_aberrance} holds for the difference $R_{A}-R_{B}$ (note signs are reversed) with 
\begin{equation}
k'=\frac{\zeta_{s}}{\zeta_{r}} \approx \sqrt{\frac{n_{0}'(n_{0}' + n_{1}')}{(n_{0}+n_{0}')(n_{0}+n_{0}'+n_{1}')}}
\end{equation}
in the place of $k$. The improvement in type-1 error rate for a SNP with aberrance in $C_{0}'$ is generally larger than the loss with the same aberrance in $C_{1}'$ (see methods), meaning that if aberrances are of similar prevalence and size in $C_{1}'$ and $C_{0}'$, method B will typically have a lower type-1 error rate than method A.

\subsubsection*{Aberrance in $C_{0}$}

Aberrance in $C_{0}$ represents a serious problem in case-control study comparison. False-positive rates are generally worse under method B, and tend to 1 as $E(z) \to \infty$. If aberrances of this type are expected to be very frequent, this may preclude use of methods B or C.

However, aberrances of this type may be best detected retrospectively by examining aberrances between control groups at SNPs declared `hits'. This procedure is already a necessary quality-control procedure in method A~\cite{wtccc07,anderson10}, as method A does not provide any control over differences between $C_{0}$ and $C_{0}'$. The number of SNPs reaching significance in the two-stage procedure is usually small enough that this examination is readily tractable. 

\subsubsection*{Aberrance in two or more cohorts}

If SNPs are aberrant in both $C_{1}$ and $C_{1}'$, or in both $C_{0}$ and $C_{0}'$, the effect on $R_{A}$ and $R_{B}$ is similar. If both cohorts are aberrant in the same direction, there is no way to differentiate the SNP from a genuine association on the basis of the genotype data alone. If cohorts are aberrant in different directions, then in both methods, the type-1 error rate is lower than for a null SNP with no aberration or aberration in only one cohort, as effect sizes for the discovery and replication cohorts are biased in opposite directions. The same typically holds if $C_{0}'$ and $C_{1}$, or $C_{0}$ and $C_{1}'$, are biased in the same direction. 

If $C_{0}'$ and $C_{1}'$ or $C_{0}$ and $C_{1}$ are both biased in the same direction, $R_{A}$ is generally lower than $R_{B}$, as $\zeta_{s} \neq 0$. Both $R_{A}$ and $R_{B}$ are bounded by $\frac{\alpha}{2}$ in this case. In addition, a systematic bias in both replication groups (or both discovery groups) is likely to be due to a known confounder, the effect of which can be removed by performing a stratified test (as is typically good practice when confounders are known). Aberrance in opposite directions leads to $R_{B}>R_{A}$ in the first case, and a scenario similar to aberrance in $C_{0}$ in the second case.

Aberrance in three or more cohorts corresponds to a chaotic scenario in which neither methods A,B, or C will reliably provide FPR control. Aberrance of this extent is typically detectable and removable using quality control procedures.


\clearpage

\bibliography{General}

\clearpage

\section*{Supplementary figures}
\clearpage

\begin{figure}[!htbp]
\begin{subfigure}[b]{\textwidth}
\begin{subfigure}[b]{0.4\textwidth}
\includegraphics[width=\textwidth]{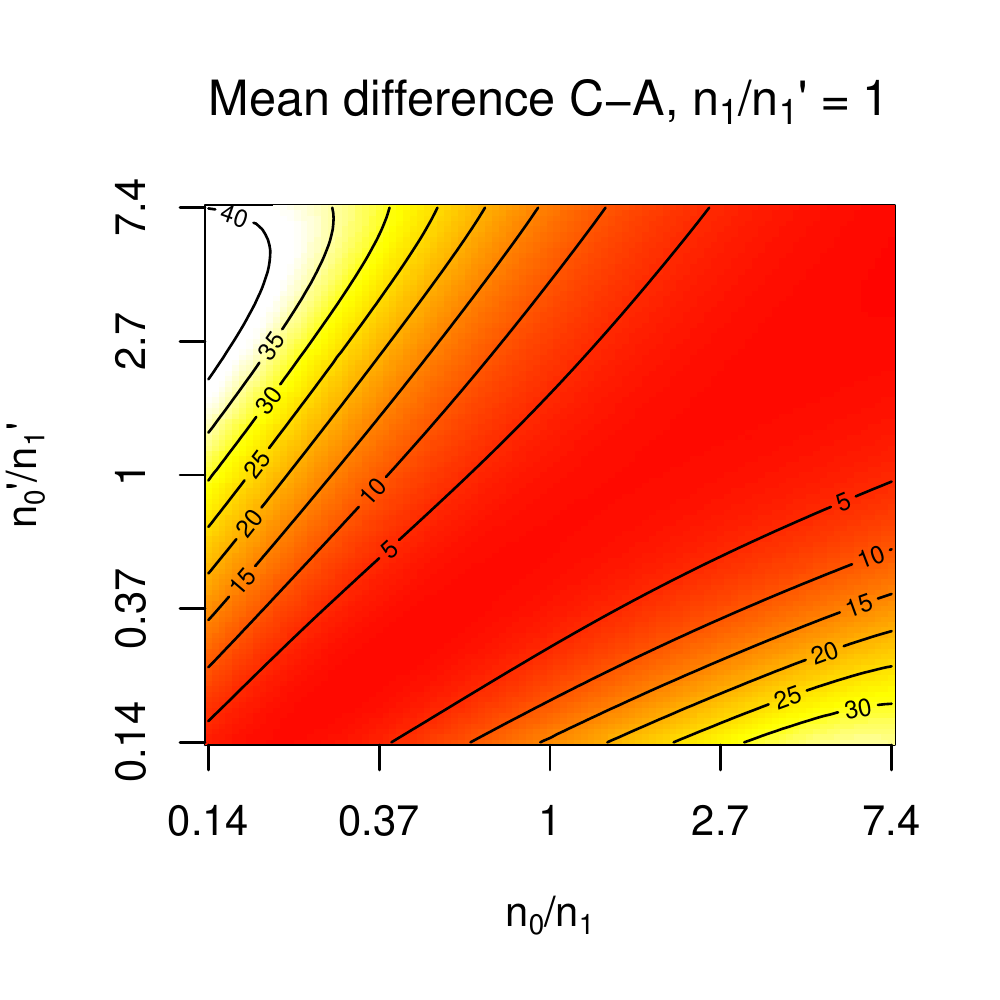}
\end{subfigure}
\begin{subfigure}[b]{0.4\textwidth}
\includegraphics[width=\textwidth]{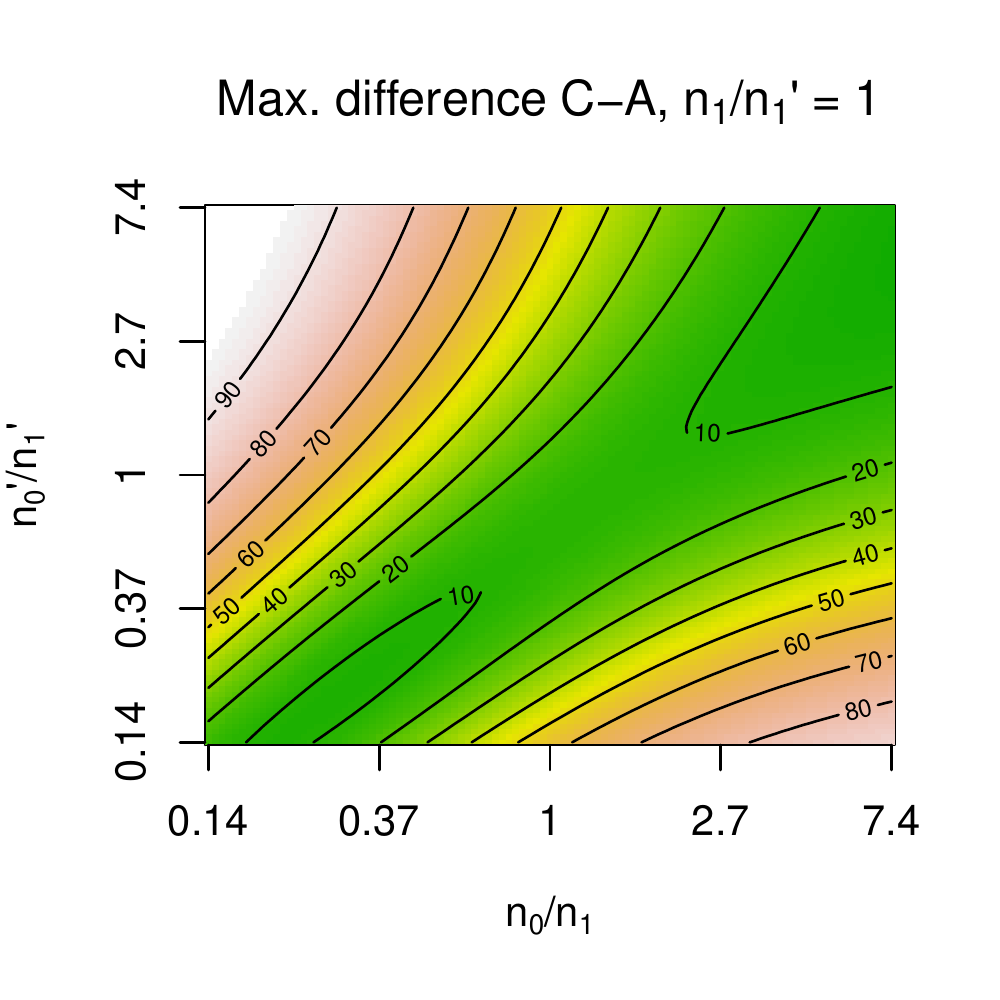}
\end{subfigure}
\begin{subfigure}[b]{0.13\textwidth}
\includegraphics[width=\textwidth]{index.pdf}
\end{subfigure}
\end{subfigure}

\begin{subfigure}[b]{\textwidth}
\begin{subfigure}[b]{0.4\textwidth}
\includegraphics[width=\textwidth]{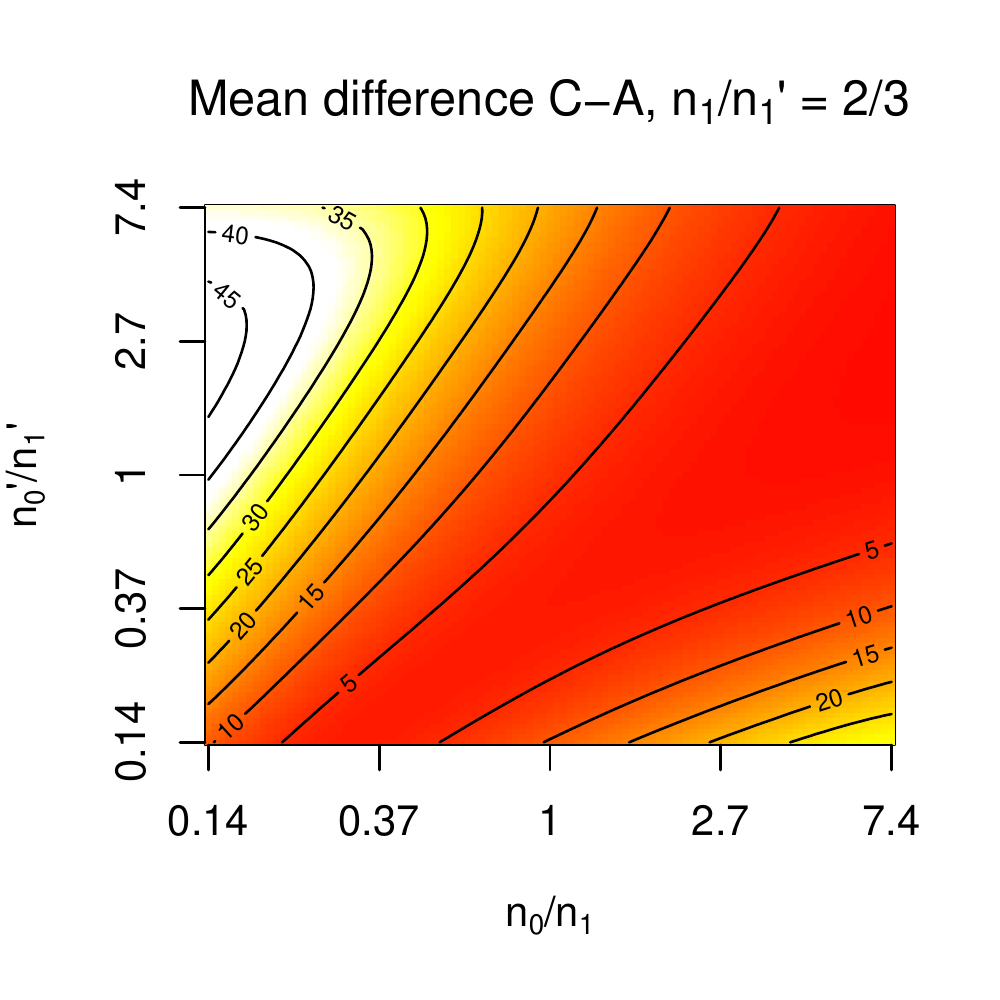}
\end{subfigure}
\begin{subfigure}[b]{0.4\textwidth}
\includegraphics[width=\textwidth]{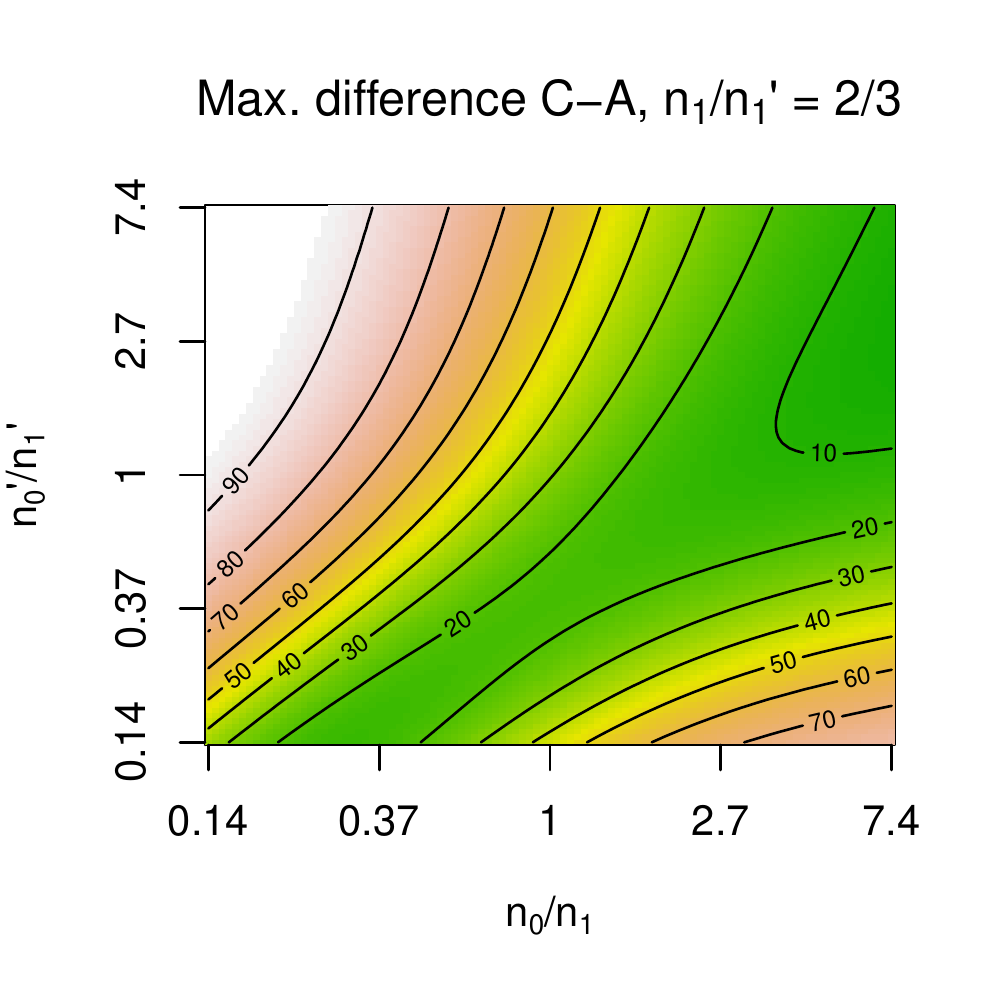}
\end{subfigure}
\end{subfigure}

\begin{subfigure}[b]{\textwidth}
\begin{subfigure}[b]{0.4\textwidth}
\includegraphics[width=\textwidth]{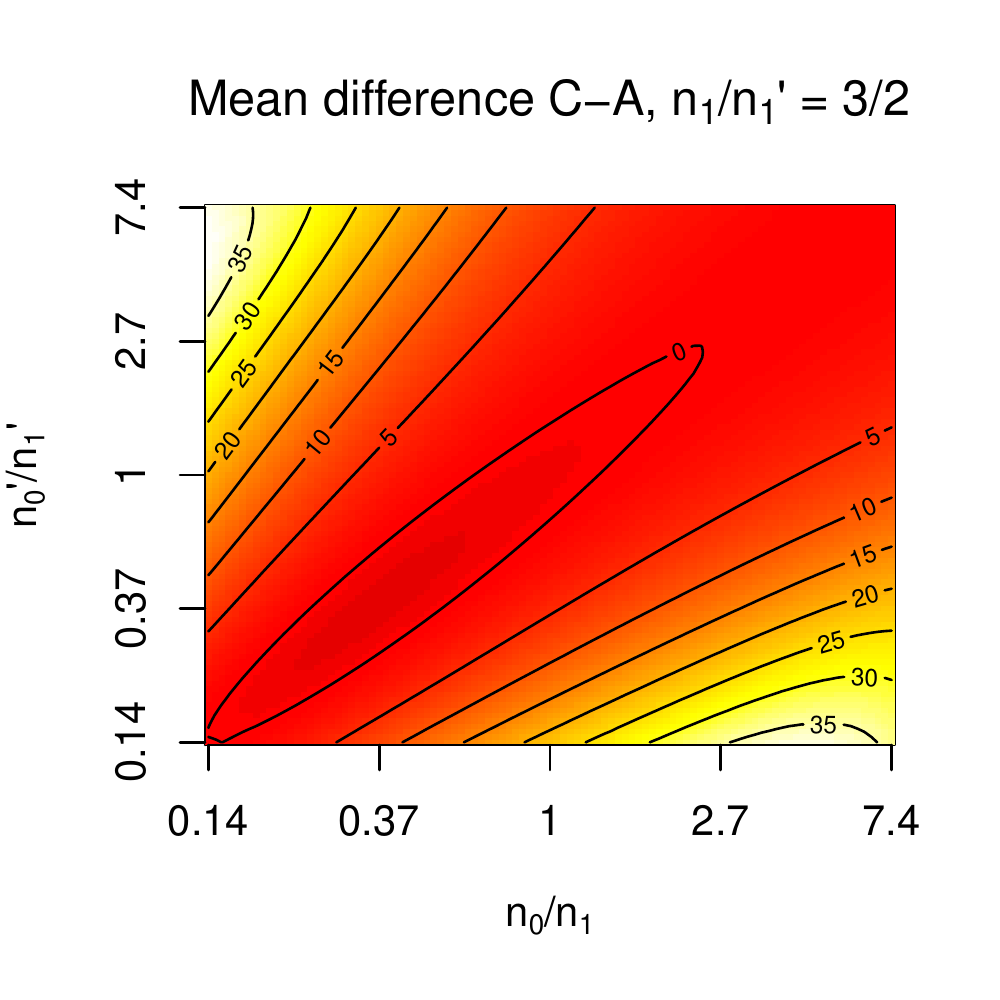}
\end{subfigure}
\begin{subfigure}[b]{0.4\textwidth}
\includegraphics[width=\textwidth]{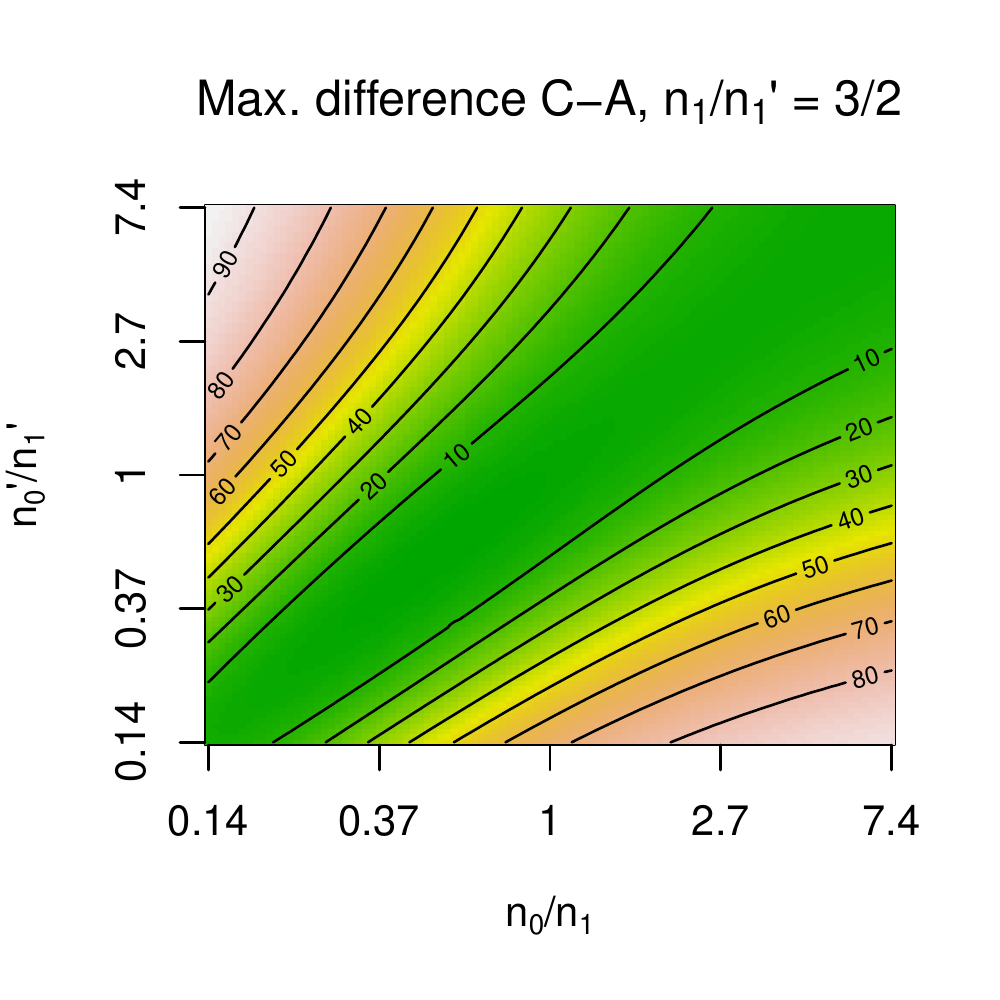}
\end{subfigure}
\end{subfigure}
\caption{Power difference (\%) between methods C and A. Mean power difference is taken as the integral of power difference between methods (see methods section) over $\mathbb{R}$ with respect to log-odds ratio. In all cases, 20 000 samples are used overall for a SNP with MAF 0.1, with cutoffs $\alpha=5 \times 10^{-6}$, $\beta=5 \times 10^{-4}$, $\gamma=5 \times 10^{-8}$. Method C is almost universally more powerful.}
\label{fig:powerAC}
\end{figure}

\begin{figure}[!htbp]
\begin{subfigure}[b]{\textwidth}
\begin{subfigure}[b]{0.4\textwidth}
\includegraphics[width=\textwidth]{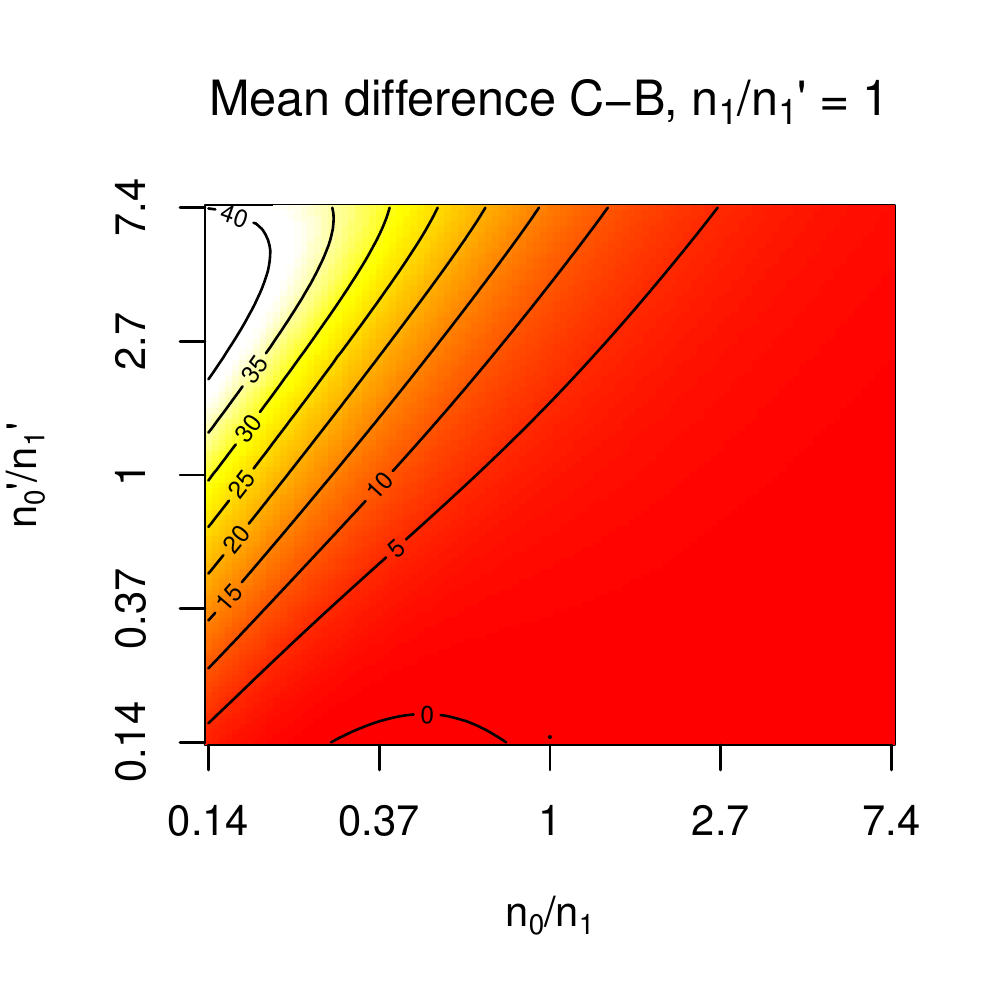}
\end{subfigure}
\begin{subfigure}[b]{0.4\textwidth}
\includegraphics[width=\textwidth]{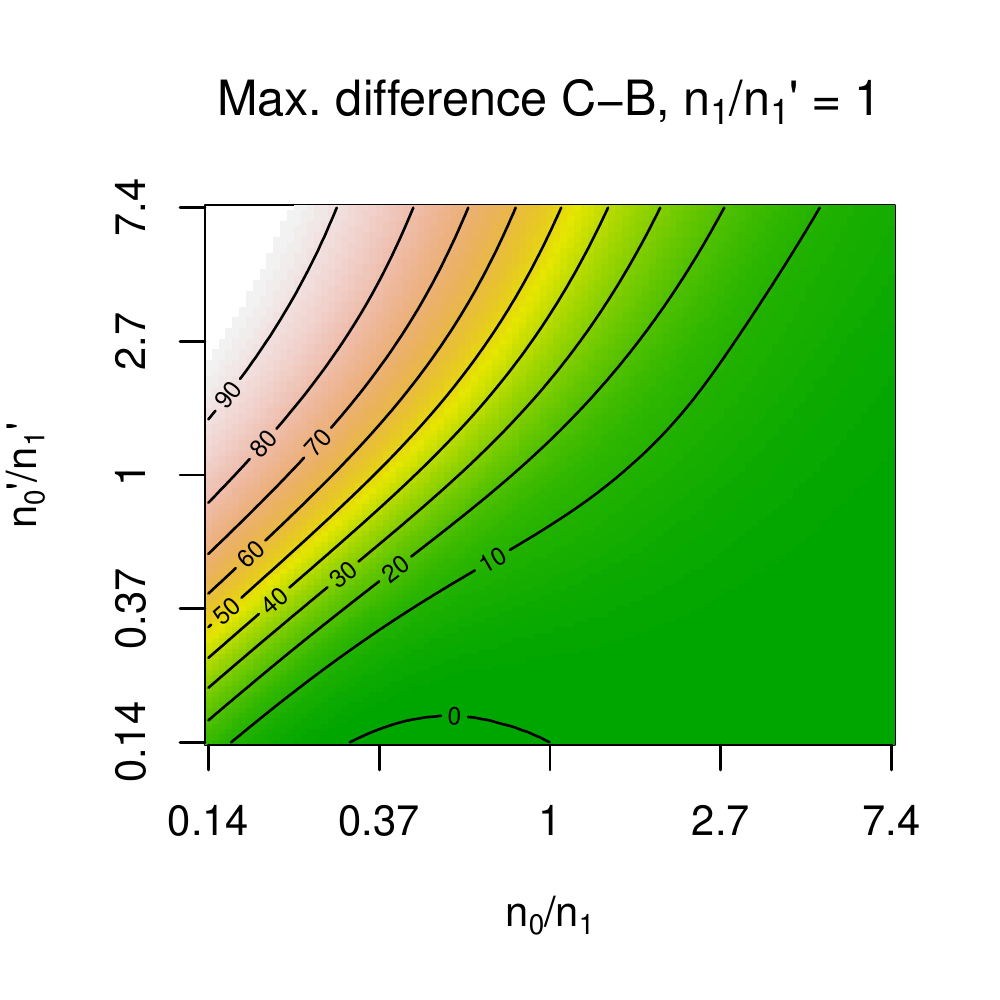}
\end{subfigure}
\begin{subfigure}[b]{0.13\textwidth}
\includegraphics[width=\textwidth]{index.pdf}
\end{subfigure}
\end{subfigure}

\begin{subfigure}[b]{\textwidth}
\begin{subfigure}[b]{0.4\textwidth}
\includegraphics[width=\textwidth]{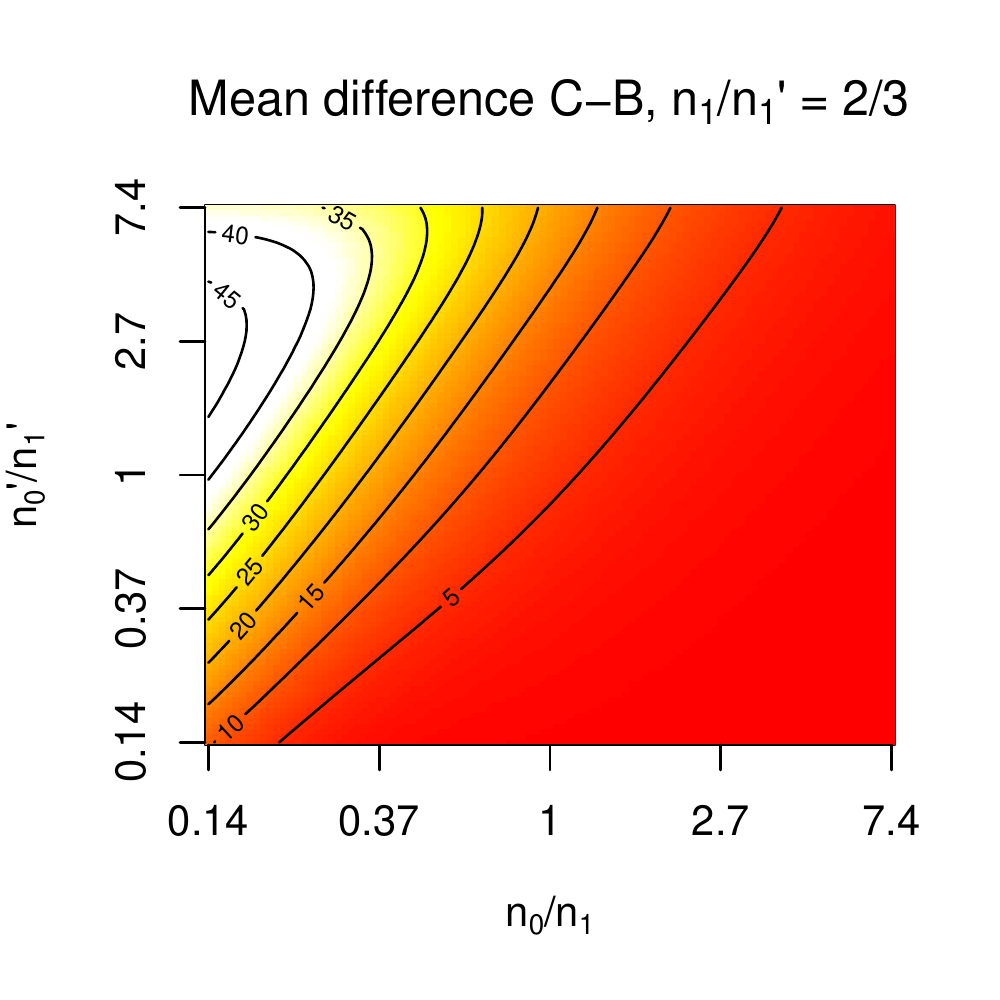}
\end{subfigure}
\begin{subfigure}[b]{0.4\textwidth}
\includegraphics[width=\textwidth]{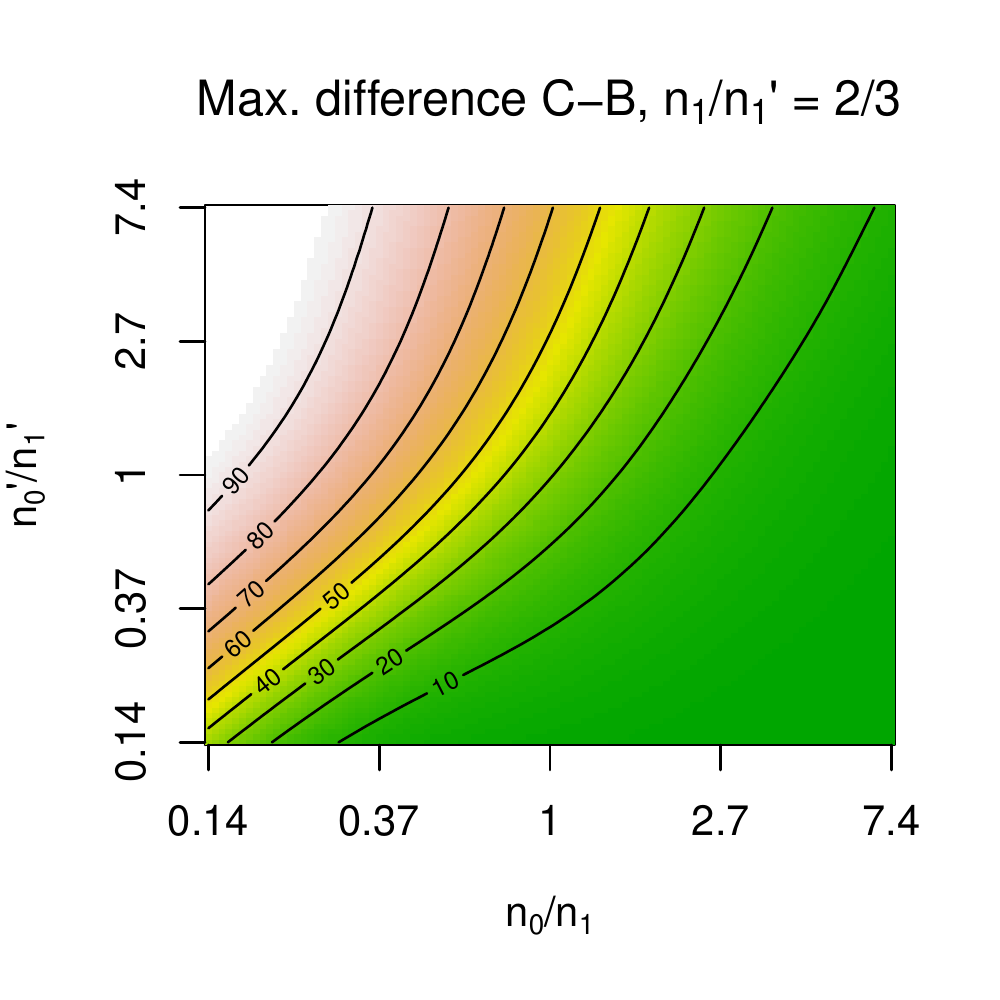}
\end{subfigure}
\end{subfigure}

\begin{subfigure}[b]{\textwidth}
\begin{subfigure}[b]{0.4\textwidth}
\includegraphics[width=\textwidth]{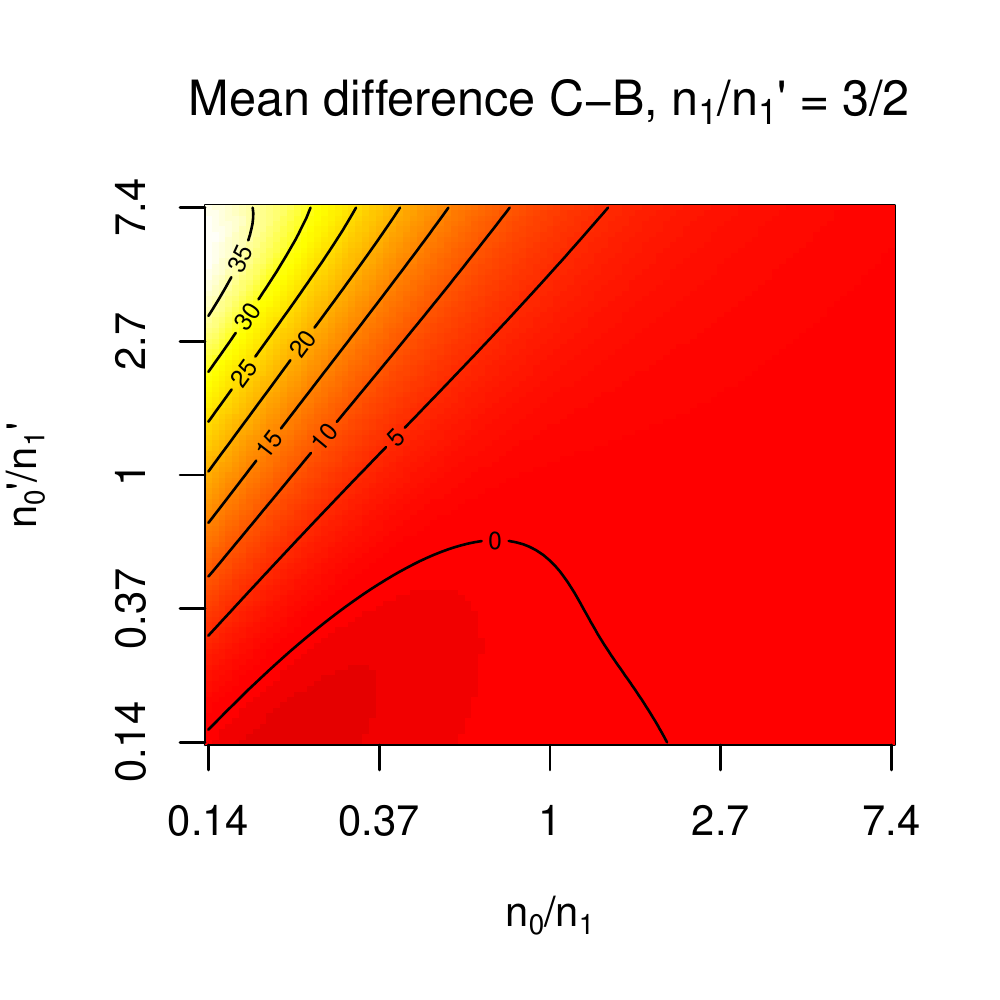}
\end{subfigure}
\begin{subfigure}[b]{0.4\textwidth}
\includegraphics[width=\textwidth]{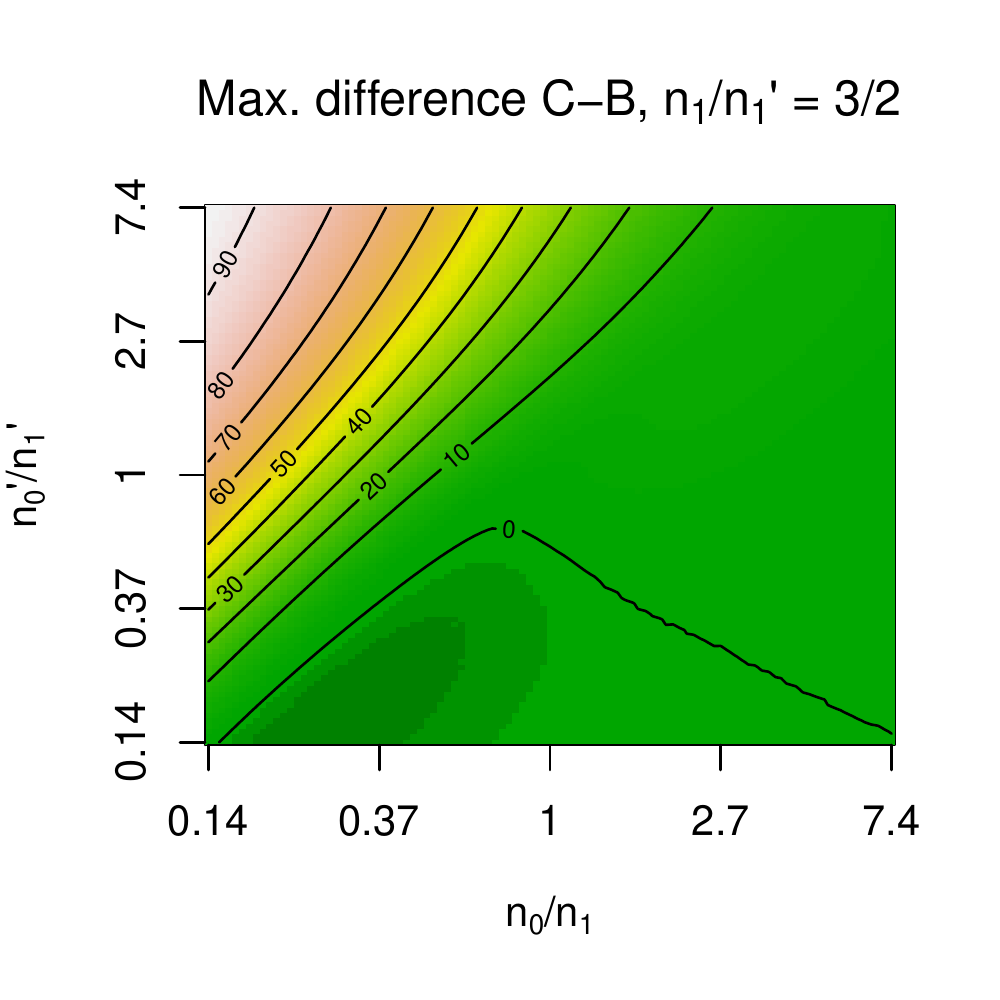}
\end{subfigure}
\end{subfigure}
\caption{Power difference (\%) between methods C and B. Mean power difference is taken as the integral of power difference between methods (see methods section) over $\mathbb{R}$ with respect to log-odds ratio. In all cases, 20 000 samples are used overall for a SNP with MAF 0.1, with cutoffs $\alpha=5 \times 10^{-6}$, $\beta=5 \times 10^{-4}$, $\gamma=5 \times 10^{-8}$.}
\label{fig:powerBC}
\end{figure}

\clearpage

\section{Appendices}
\subsection{Covariance between Z scores due to shared samples}
\label{apx:zcor}

The matching of type-1 error rates between methods relies on establishing the covariance between triples of z-scores under $H_{0}^\cap$. The covariance can be readily approximated when z-scores are assumed to be derived from tests of equality of binomial random variables $m_{i}$. Z-scores comparing proportions without using additional information (such as strata or covariates) which are monotonic to allelic difference and asymptotically have $N(0,1)$ distribution must be asymptotically equivalent to those derived from binomial comparisons, so this assumption is reasonable. 

If strata or covariates are used, either an assumption must be made that computed z-scores are well-approximated by comparisons of binomial proportions, or correlations must be approximated allowing for the covariate or strata structure. This is sometimes tractable analytically, but can also be estimated empirically either by using known non-associated variants or by simulating variants with the same covariate structure.

The presence of strata or covariates also affects the values $\zeta_i$, and if the effects of covariates are large, the approximations in equations~\ref{eq:zetadef} in the main paper may be poor. Values $\zeta_i$ can be estimated as functions of allelic differences by simulating variants with the same covariate structure.

\subsubsection{No covariates or stratification}

Assume study $i$ and $j$ have $n_{0}^{i}$, $n_{0}^{j}$ controls and $n_{1}^{i}$, $n_{1}^{j}$ cases respectively, of which $n_{0}^{ij}$ controls and $n_{1}^{ij}$ cases are shared between both studies. Let $m_{0},m_{1},m_{0}',m_{1}'$, denote the observed allele frequencies of a SNP in the respective cohort, and $\mu_{0},\mu_{1},\mu_{0}',\mu_{1}'$ the expected allele frequency. 

If no strata or covariates are used in the calculation of summary statistics, $z$ scores $z_{d}, z_{r}, z_{s}, z_{m}$ are asymptotically proportional to the allelic differences $m_{1}-m_{0}$, $m_{1}'-m_{0}'$, $m_{1}'-\frac{m_{0}'n_{0}' + m_{0}n_{0}}{n_{0}+n_{0}'}$, $\frac{m_{1}n_{1}+m_{1}'n_{1}'}{n_{1}+n_{1}'}-\frac{m_{0}'n_{0}' + m_{0}n_{0}}{n_{0}+n_{0}'}$ respectively, since $z$ scores are monotonic with allelic differences and allelic differences are asymptotically normal.
Since $m_{0},m_{1},m_{0}',m_{1}'$ are independent and asymptotically normal the multivariate random variables $(z_{d}, z_{r}, z_{m})$ and $(z_{d}, z_{s}, z_{m})$ have multivariate normal distributions. 

For studies $i$ on $n_{0i}$, $n_{1i}$ controls and cases and $j$ on $n_{0j}$, $n_{1j}$ controls and cases in which $n_{0ij}$ and $n_{1ij}$ controls and cases are shared between studies, the correlation between the observed allelic differences $m_{1i}-m_{0i}$, $m_{1j}-m_{0j}$ is given by 
\begin{equation}
cor(m_{1i}-m_{0i},m_{1j}-m_{0j})=\frac{n_{0i}n_{0j}n_{1ij} + n_{1i}n_{1j}n_{0ij}}{n_{0i}n_{0j}n_{1i}n_{1j}\sqrt{\frac{1}{n_{0i}}+\frac{1}{n_{0j}}} \sqrt{\frac{1}{n_{0j}}+\frac{1}{n_{1j}}}} \label{eq:covdef}
\end{equation}
This holds under $H_{0}^{\cap}$ and approximately holds in general.
Expressions for $\rho_{ds}$, $\rho_{dm}$, $\rho_{rm}$ and $\rho_{sm}$ may be derived in terms of $n_{0}$, $n_{1}$, $n_{0}'$, and $n_{1}'$. Specifically
\begin{align}
det(\Sigma_{A}) &= 1  - \rho_{dm}^{2} - \rho_{rm}^{2} \nonumber \\
&= \frac{(n_{0} n_{1}' - n_{0}'n_{1})^{2}}{(n_{0}+n_{0}')(n_{1}+n_{1}')(n_{0}+n_{1})(n_{0}'+n_{1}')} \label{eq:detsigmaA} \\
det(\Sigma_{B}) &= 1-\rho_{dm}^{2} - \rho_{ds}^{2} - \rho_{sm}^{2} + 2 \rho_{dm}\rho_{ds}\rho_{sm} \nonumber \\ 
&= \frac{n_{0}' n_{1}^{2}}{(n_{0}+n_{1})(n_{0}+n_{0}'+n_{1}')(n_{1}+n_{1}')}
\label{eq:detsigmaB}
\end{align}
so $\Sigma_{A}$ is singular if $\frac{n_{0}}{n_{1}}=\frac{n_{0}'}{n_{1}'}$, and $\Sigma_{B}$ if $n_{0}'n_{1}=0$.





%

\subsubsection{Z scores with stratification}

If computation of Z scores is performed with correction for strata or covariates, formula~\ref{eq:covdef} will not asymptotically hold and may be a poor approximation to the true covariance. The true covariance can be computed in some cases.

If samples are divided into strata $1,2,...s$, and $n_{pq}^{r}$, $m_{pq}^{r}$, $\mu_{pq}^{r}$ denote the number of samples and  observed and expected minor allele frequencies in cohort $p$, study $q$, stratum $r$ respectively, then the $z$ score $z_{i}$ for study $q=i$ is asymptotically given by
\begin{equation}
z_{i} = \sum_{r \in 1..s} \alpha_{ir} \left( m_{1i}^{r}-m_{0i}^{r} \right)
\end{equation}
for positive values $\alpha_{ir}$ depending on the values $n_{pi}^{r}$. If the Cochran-Mantel-Hanszel test is used, then 
\begin{equation}
\alpha_{ir} \propto \frac{n_{0i}^{r}n_{1i}^{r}}{n_{0i}^{r} + n_{1i}^{r}}
\end{equation}
Suppose that $n_{0ij}^{r}$ controls and $n_{1ij}^{r}$ cases are shared between studies $i$ and $j$ in stratum $r$. Since the values $m_{pq}^{r}$ are dependent only within the same values of $p$ and $r$, the correlation between $z_{i}$ and $z_{j}$ under the null hypothesis $\mu_{0i}^{r} \equiv \mu_{1i}^{r}, \mu_{0i}^{r} \equiv \mu_{1i}^{r}$ is given by
\begin{align}
cor(z_{i},z_{j}) &= \frac{\sum
\alpha_{ir}\alpha_{jr}cov(m_{1i}^{r}-m_{0i}^{r},m_{1j}^{r}-m_{0j}^{r})}{\sqrt{var(m_{1i}^{r}-m_{0i}^{r})var(m_{1j}^{r}-m_{0j}^{r})}} \nonumber \\
&\approx \frac{\sum_{r \in 1..s} \alpha
\left(\frac{n_{0ij}^{r}}{n_{0i}^{r}n_{0j}^{r}}+\frac{n_{1ij}^{r}}{n_{1i}^{r}n_{1j}^{r}}\right) }{\sqrt{\left(\sum\alpha_{ir}^{2}\left(\frac{1}{n_{0i}^{r}} + \frac{1}{n_{1i}^{r}}\right)\right)\left(\sum\alpha_{jr}^{2}\left(\frac{1}{n_{0j}^{r}} + \frac{1}{n_{1j}^{r}}\right)\right)}}
\end{align}
where all sums are over the values of $r \in 1..s$ for which the relevant values of $n_{pq}^{r}$ are positive.

\subsubsection{Z scores with covariates}

If $z$ scores are computed adjusting for one or more covariates, the estimation of correlation is more difficult. 

Assume that in a case population $C_1$ and a control population $C_0$ the values of some covariate(s) $x$ have different known distributions $f_1$, $f_0$, and that genotypes $g$ at some SNP of interest may vary with $x$. We will assume the populations are large and that $f_1$, $f_0$, and $E(g|x)$ are continuous functions of $x$. 



Let $g_{p}^k$ denote the genotype of individual $k$ in cohort $p$ ($p \in C_0, C_1$) and $x_p^k$ denote covariate value(s), where $g_p^k$ is an observation of a random variable $g$. 
An idealised $z$-score testing association of $g$ with case/control status should be monotonic with each $g_p^k$ and have expectation 0 if $x$ is independent of case/control status, whatever the form of the function $E(g|x)$. Because individual genotypes are assumed to be independent between individuals, cross-terms of the form $\prod_i g_i$ should carry no additional information from singleton genotypes. We thus assume that $z$ can thus be decomposed into a weighted linear sum of 
individual genotypes:
\begin{equation}
z \propto \frac{1}{|C_{1}|}\sum_{k \in C_{1}} c_{1}^{k} g_1^{k} - \frac{1}{|C_{0}|}\sum_{k \in C_{0}} c_0^k g_0^k \label{eq:genzdef}
\end{equation}
where the (positive) values $c_1^k, c_0^k$ depend only on the values $x_1^\cdot,x_0^\cdot$; that is, not on the relationship between $g$ and $x$, and the constant of proportionality depends on only on the observed allele frequency. Let function $c_{0}(x),c_{1}(x)$ denote the values of $c_{i}$ corresponding to covariate value(s) $x$ in $C_0$, $C_1$. 


For a null SNP, $E(g|x)$ is independent of case/control status, but may take any (continuous) form. We have
\begin{align}
E(z) &\propto E \left(\frac{1}{|C_{1}|}\sum_{i \in C_{1}} c_{i} g_{i} - \frac{1}{|C_{0}|}\sum_{i \in C_{0}} c_i g_i \right) \nonumber \\
\lim_{|C_{0}|,|C_{1}| \to \infty} E(z) &\propto \int c_{1}(x)f_{1}(x)E(g|x) dx - \int c_{0}(x)f_{0}(x)E(g|x) \nonumber \\
&\propto \int \left( c_{1}(x)f_{1}(x) -  c_{0}(x)f_{0}(x) \right) E(g|x) dx
\end{align}
From a standard result from the calculus of variations, this implies that
\begin{equation}
c_1 (x) f_1 (x) - c_0 (x) f_0 (x) \equiv 0
\implies c_{1}(x) \propto \frac{f(x)}{f_{1}(x)}, c_{0}(x) \propto \frac{f(x)}{f_{0}(x)} \label{eq:commondist}
\end{equation}
for some function $f$, so the values $c_1^k, c_0^k$ effectively reweight the contribution of individuals to a common density $f(x)$ across $x$. The procedure of weighting observation $k$ in this way is analogous to a limiting case of stratification, in which weights are defined by the frequency of stratum $r$ (see above). For a constant allelic difference across the range of $x$, the best common distribution to `map to' does not depend on the relationship between $g$ and $x$, and hence the best values of $c_i$ should be constant for all functions $E(g|x)$.

Let $z_{q}$ denote a z-score for study $q$; $n_{pq}$, $f_{pq}=f_{pq}(x)$ and $C_{pq}$ denote the number of samples, density function of $x$, and set of samples in cohort $p$, study $q$; $g_{pq}^{i}$ and $c_{pq}^{k}$ denote the normalised genotype of sample $k$ in cohort $p$, study $q$ and its coefficient in $z_{q}$; $n_{0s}$, $n_{1s}$, $f_{0s}$, $f_{1s}$ and $C_{0}^s$, $C_{1}^{s}$ the number of shared controls/cases between studies, the density of $x$ amongst the shared samples, and the sets of shared samples; and $f_{q}$ the common density function to which cases and controls are weighted in study $q$ (equation~\ref{eq:commondist}). Then
\begin{align}
cov(z_{i},z_{j}) 
&\approx \frac{\frac{1}{n_{0i}n_{0j}} \sum_{k \in C_{0}^{s}} c_{0i}^k c_{0j}^{k} +  \frac{1}{n_{1i}n_{1j}} \sum_{k \in C_{1}^{s}} c_{1i}^k c_{1j}^{k}}{
\sqrt{\frac{1}{n_{1i}^2}\sum_{k \in C_{1i}} ( c_{1i}^k )^2 + \frac{1}{n_{0i}^2}\sum_{k \in C_{0i}} ( c_{0i}^{k} )^2 }
\sqrt{\frac{1}{n_{1j}^2}\sum_{k \in C_{1j}} ( c_{1j}^k )^2 + \frac{1}{n_{0j}^2}\sum_{k \in C_{0j}} ( c_{0j}^{k} )^2 }
} \label{eq:scordef} \\
&\to \frac{\frac{n_{0s}}{n_{0i}n_{0j}}\int f_{0s}(x)\frac{f_{i}(x)f_{j}(x)}{f_{0i}(x)f_{0j}(x)} dx + \frac{n_{1s}}{n_{1i}n_{1j}}\int f_{1s}(x) \frac{f_{i}(x)f_{j}(x)}{f_{1i}(x)f_{1j}(x) dx}}{\sqrt{\frac{1}{n_{0i}} \int \frac{f_i (x)^2}{f_{0i}(x)} dx + \frac{1}{n_{1i}} \int \frac{f_{i}(x)^2}{f_{1i}(x)} dx}\sqrt{\frac{1}{n_{0j}} \int \frac{f_j (x)^2}{f_{0j}(x)} dx + \frac{1}{n_{1j}} \int \frac{f_{j}(x)^2}{f_{1j}(x)} dx}}
\end{align}
with integrals over the domain of $x$, and the limit as sample sizes tend to infinity while ratios between them remain bounded.

Logistic regression models with continuous covariates can only model simple (generally linear) relationships between $c_i$ and $x_i$, and property~\ref{eq:commondist} may not hold. If the values $c_{pq}^k$ are known, the correlation can be determined using equation~\ref{eq:scordef}. If not, some methods for estimating correlation are outlined below.



\subsubsection{Practical estimation of covariance}

Although the asymptotic correlation between $z$ scores may be intractable, as long as the $z$ score permits an expansion of the form~\ref{eq:genzdef}, the correlation will be nearly invariant with allele frequency and change only minimally for SNPs associated with the covariate.

In practical terms, one method to estimate the correlation between $z$ scores is to simply use the sample correlation at a set of variants presumed to be not associated with the main trait of interest. This approach may be unreliable and have limited power due to the difficulty of identifying such variants

Another option is to permute existing genotypes without permuting covariates, and compute correlation between resultant $z$ scores. This has the disadvantage that it is difficult to permute whilst maintaining potential relationships between genotypes and confounders.

Since the correlation should only depend on the sample sizes and structure of covariate distributions, a more convenient and powerful method is to simply simulate sets of genotypes unassociated with the trait, but potentially associated with covariates in a range of different ways, and compute correlation between the resultant $z$ scores. Given the shortcomings of standard methods such as logistic regression in fully accounting for covariate effects, this is an advisable procedure in any analysis including covariates.

All results in the main paper which require conditions on sample sizes are only approximate when using studies with stratification or covariates, with the approximation worsening with greater differences in covariate values between groups and lower effective sample sizes.

\subsection{Properties of $\beta^*$}

\subsubsection{Asymptotic properties of $\beta^*$}
\label{apx:infbetastar}

In this appendix, an asymptotic approximation is established for $\beta^*$
and it is shown that $\beta^{*}>\beta$ for all $n_{0}^{i}$, $n_{0}^{j}$, $n_{1}^{i}$, $n_{1}^{j}$, $z_{\alpha}$, $z_{\gamma}$. Define $\Sigma_{A}$ and $\Sigma_{B}$ as per equations \ref{eq:sigcdef} in the main paper, and note that $\Sigma_{A}$ and $\Sigma_{B}$ only differ in their middle row/column. Further define 
\begin{equation}
\Sigma_{dm}=var\left( (z_{d} \, z_{m})^t | H_{0}^{\cup}\right)=\begin{pmatrix} 1 & \rho_{dm} \\ \rho_{dm} & 1 \end{pmatrix}
\end{equation}
Let $(z_{\alpha}' \, z_{\gamma}')$ be the point in $\{ z_{d}>z_{\alpha}, z_{m}>z_{\gamma} \}$ at minimal Mahalanobis distance from the origin with respect to $\Sigma_{dm}$ (ie, minimal $(z_{d} \, z_{m})\Sigma_{dm}^{-1} (z_{d} \, z_{m})^t$). Then for $z_{\gamma}' - \rho_{dm}z_{\alpha}'$ held constant, we have
\begin{equation}
\lim_{z_{\gamma}' \to \infty / z_{\alpha}' \to \infty} \frac{\sqrt{|\Sigma_{A}|}\left((\rho_{ds}\rho_{dm}-\rho_{sm})z_{\gamma}' + (\rho_{dm}\rho_{sm} - \rho_{ds})z_{\alpha}' + |\Sigma_{dm}|z_{\beta^*}\right)}{\sqrt{|\Sigma_{B}|}\left(-\rho_{rm}z_{\gamma}' + \rho_{dm}\rho_{rm}z_{\alpha}' + |\Sigma_{dm}|z_{\beta}\right)}=1
\end{equation}
Specifically, for  $\beta^*$ defined as per equation~\ref{eq:betastardef2}, we have
\begin{equation}
\lim_{\alpha \to 0} \frac{z_{\beta^{*}}}{\sqrt{1-\rho_{ds}^{2}}z_{\beta} + \rho_{ds} z_{\alpha}}=1 \label{eq:infbetastar2}
\end{equation}
and $z_{\beta^{*}} > max(\beta,\sqrt{1-\rho_{ds}^{2}}z_{\beta} + \rho z_{\alpha})$ for all $z_{\alpha}$. Firstly the following lemma and corollary are established:

\begin{lemma}
\label{lemma:genlimit}
Let $\Sigma$ be a positive definite matrix of dimension $N$, $\mathbf{x}$ be the vector $(x_{1} \, x_{2} ... x_{n})^{t}$, $\mathbf{A_{1}}$, $\mathbf{A_{0}}$, and $\mathbf{Z}=(z_{1} \, z_{2} ... z_{n})^{t}$ constant vectors of dimension $N$ with $\mathbf{A_{1}} \ne \mathbf{A_{0}} \neq 0$, $C_{0}$ a constant, and $R$ the (closed) region $x_{1}\geq z_{1}, x_{2} \geq z_{2},,,x_{N} \geq
z_{N}$.

Define $C$ as the (unique) value satisfying
\begin{equation}
\int_{R} 
e^{-\frac{1}{2} \mathbf{x^{t}} \Sigma^{-1} \mathbf{x}}  \left(\Phi(\mathbf{A_{1}^{t} x}+C) - \Phi(\mathbf{A_{0}^{t}x} + C_{0})  \right) dx_{1} dx_{2}...dx_{N} = 0 \label{eq:lemma2}
\end{equation}
Denote $\mathbf{y}=(y_{1} \, y_{2} ... y_{N})$ as the point in R at minimal Mahalanobis distance $M(\mathbf{y})$ from the origin with respect to $\Sigma$ (usually, $\mathbf{y}=\mathbf{Z}$). Consider all regions $R$ for which the corresponding value of $\mathbf{y}$ lies on the hyperplane $\mathbf{A_{0}^{t}y} + C_{0}'=0$, $C_{0}' \ne C_{0}$. Then
\begin{equation}
\lim_{M(y) \to \infty | \mathbf{A_{0}^{t}y} + C_{0}'=0} \frac{\mathbf{A_{1}^{t} y}+C}{\mathbf{A_{0}^{t}y}+C_{0}} = \lim_{M(y) \to \infty | \mathbf{A_{0}^{t}y} + C_{0}'=0} \frac{\mathbf{A_{1}^{t} y}+C}{C_{0}- C_{0}'} = 1 
\end{equation}
\end{lemma}

\begin{proof}
The value $C$ is unique since the function $\Phi(\mathbf{A_{1}^{t} x}+C)$ is continuous and monotonically increasing in $C$ for all $\mathbf{x}$, and hence so is the integrand (and integral).

We proceed from the formal definition of a limit
\begin{equation}
\forall \, \epsilon>0  \, \, \exists Y \, | \, \left( M(y) > Y \implies \left| \frac{\mathbf{A_{1}^{t} y}+C}{\mathbf{A_{0}^{t} y} + C_{0}} - 1 \right| < \epsilon \right)
\end{equation}
Because $\mathbf{A_{0}^{t}y}+C_{0}'=0$, the right-hand side is equivalent to
\begin{equation}
(1-\epsilon)(C_{0}-C_{0}') - \mathbf{A_{1}^{t}y} \leq C \leq (1+\epsilon) (C_{0} -C_{0}') - \mathbf{A_{1}^{t}y} \label{eq:cinterval}
\end{equation}
We will show that there exists $Y$ such that $M(y) > Y$  implies that when $C$ takes values at the endpoints of the interval in the integral \ref{eq:cinterval}, the integral \ref{eq:lemma2} takes different signs. Since the integral is increasing in $C$ and must be 0, $C$ must lie in the interval in \ref{eq:cinterval} for $M(y)>Y$.

If $C$ takes the upper value, then at $\mathbf{x}=\mathbf{y}$, the value of the integrand is
\begin{equation}
e^{-\frac{1}{2} M(\mathbf{y})}  \left(\Phi((1+\epsilon) (C_{0}-C_{0}')) - \Phi(C_{0}-C_{0}')  \right) 
 \end{equation}
%
%
%
the sign of which depends on the sign of $C_{0}-C_{0}'$. We shall assume it is positive (with analogous arguments if it is negative). Because $\epsilon>0$, point $\mathbf{y}$ does not lie on the hyperplane $(\mathbf{A_{1}^{t}}-\mathbf{A_{0}^{t}})\mathbf{x}+(1+\epsilon)(C_{0}-C_{0}') - \mathbf{A_{1}^{t}y} - C_{0}=0$ (on which the integrand of~\ref{eq:lemma2} is 0). The distance from $\mathbf{y}$ to the hyperplane is given by
\begin{align}
D &= \frac{|(\mathbf{A_{1}^{t}}-\mathbf{A_{0}^{t}})\mathbf{y}+(1+\epsilon)(C_{0}-C_{0}') - \mathbf{A_{1}^{t}y} - C_{0}|}{||\mathbf{A_{1}^{t}}-\mathbf{A_{0}^{t}}||} \nonumber \\
&=\frac{|(1+\epsilon)(C_{0}-C_{0}') - C_{0}-C_{0}'|}{||\mathbf{A_{1}^{t}}-\mathbf{A_{0}^{t}}||}
\end{align}
which is independent of $\mathbf{y}$. Consider a hypersphere centred at $\mathbf{y}$ of radius $d<D$. Each point in the hypersphere can be expressed as $\mathbf{y} + \mathbf{\kappa}$ with $|\mathbf{\kappa}| \leq d$, so within the hypersphere we have
\begin{align}
\Phi(\mathbf{A_{1}^{t} x}+C) - \Phi(\mathbf{A_{0}^{t}x} + C_{0}) &= \Phi(\mathbf{A_{1}^{t}}(\mathbf{y}+\mathbf{\kappa})+(1+\epsilon) (C_{0} -C_{0}') - \mathbf{A_{1}^{t}y}) \nonumber \\
&\phantom{=} - \Phi(\mathbf{A_{0}^{t}}(\mathbf{y}+\mathbf{\kappa}) + C_{0}) \nonumber \\
 &= \Phi \left( (1+\epsilon)(C_{0}-C_{0}') + \mathbf{A_{1}^{t} \kappa} \right) \nonumber \\
&\phantom{=} + \Phi \left((C_{0}-C_{0}') + \mathbf{A_{0}^{t} \kappa} \right) \nonumber \\
&\geq  \Phi \left( (1+\epsilon)(C_{0}-C_{0}') + |\mathbf{A_{1}^{t}}| d \right) \nonumber \\
&\phantom{=} + \Phi \left((C_{0}-C_{0}') - |\mathbf{A_{0}^{t}}| d \right)
\end{align}
Thus $d$ can be chosen independently of $\mathbf{y}$ such that $\Phi(\mathbf{A_{1}^{t} x}+C) - \Phi(\mathbf{A_{0}^{t}x} + C_{0})$ is bounded below in the hypersphere by a constant $X$ also independent of $\mathbf{y}$. The function $\Phi(\mathbf{A_{1}^{t} x}+C) - \Phi(\mathbf{A_{0}^{t}x} + C_{0})$ is obviously bounded by $\pm 2$. Let $R'$ be the intersection of $R$ and the hypersphere. The integral~\ref{eq:lemma2} now satisfies
\begin{align}
&\phantom{=} \int_{R} 
e^{-\frac{1}{2} \mathbf{x^{t}} \Sigma^{-1} \mathbf{x}}  \left(\Phi(\mathbf{A_{1}^{t} x}+C) - \Phi(\mathbf{A_{0}^{t}x} + C_{0})  \right) dx_{1} dx_{2}...dx_{N} \nonumber \\
&= \int_{R'} 
e^{-\frac{1}{2} \mathbf{x^{t}} \Sigma^{-1} \mathbf{x}}  \left(\Phi(\mathbf{A_{1}^{t} x}+C) - \Phi(\mathbf{A_{0}^{t}x} + C_{0})  \right) dx_{1} dx_{2}...dx_{N} \nonumber \\
&\phantom{=} + \int_{R \setminus R'} 
e^{-\frac{1}{2} \mathbf{x^{t}} \Sigma^{-1} \mathbf{x}}  \left(\Phi(\mathbf{A_{1}^{t} x}+C) - \Phi(\mathbf{A_{0}^{t}x} + C_{0})  \right) dx_{1} dx_{2}...dx_{N} \nonumber \\
&> X  \int_{R'} 
e^{-\frac{1}{2} \mathbf{x^{t}} \Sigma^{-1} \mathbf{x}} dx_{1} dx_{2}...dx_{N} \nonumber \\
&\phantom{=} - 2  \int_{R \setminus R'} 
e^{-\frac{1}{2} \mathbf{x^{t}} \Sigma^{-1} \mathbf{x}} dx_{1} dx_{2}...dx_{N} \label{eq:intbound}
\end{align}
Because $d$ (the radius of the hypersphere) does not depend on $\mathbf{y}$, by the properties of the Gaussian integral a value $M_{+}$ can be chosen such that $M(y)>M_{+}$ implies that the ratio
\begin{equation}
\frac{ \int_{R'} 
e^{-\frac{1}{2} \mathbf{x^{t}} \Sigma^{-1} \mathbf{x}} dx_{1} dx_{2}...dx_{N}}{\int_{R \setminus R'} 
e^{-\frac{1}{2} \mathbf{x^{t}} \Sigma^{-1} \mathbf{x}} dx_{1} dx_{2}...dx_{N} }
\end{equation}
is arbitrarily large (namely, $>2/X$), and hence integral~\ref{eq:intbound} is positive. In a similar way, a value $M_{-}$ can be chosen such that if $C$ takes the lower value of interval \ref{eq:cinterval}, the integral is negative for $M(y)>M_{-}$. For $M(y)>max(M_{+},M_{-})$, the value of $C$ satisfying equation \ref{eq:lemma2} lies within the interval ~\ref{eq:cinterval}, and the limit is established.

\end{proof}

\begin{corollary}
\label{corollary:ylimit}
Given $b,c,y \in \mathbb{R}^{+}$, define $a$ such that
\begin{equation}
\int_{y}^{\infty} e^{-\frac{x^2}{2}} \left(\Phi(c) - \Phi(a-bx)  \right) dx = 0 \label{eq:lemma1}
\end{equation}
then
\begin{equation}
\lim_{y \to \infty} \frac{a}{b y + c} = 1 
\end{equation}
and $a> by + c \, \, \forall \, y$ 
\end{corollary}

\begin{proof}
We note firstly that the function $\Phi(c) - \Phi(a-bx)$ is increasing for all $x$. If the integral is 0, the (smooth) integrand must cross 0 at some finite $x \in (y,\infty)$, and hence its value at $x=y$ must be negative. As $\Phi$ is increasing, we have $\Phi(a-by) > \Phi(c) \implies a > by +c$

The proof of the limit proceeds in a similar way to the proof of the lemma above.

\end{proof}


Now (recalling definition~\ref{eq:betastardef} in the main paper)
\begin{align}
&\int_{z_{\alpha}}^{\infty} \int_{z_{\gamma}}^{\infty}  \int_{z_{\beta^{*}}}^{\infty} N_{\Sigma_{B}} \left( (z_{d} \, z_{s} \, z_{m} )^t\right) dz_{s} dz_{m} dz_{d} \nonumber \\
&= \int_{z_{\alpha}}^{\infty}  \int_{z_{\gamma}}^{\infty} \int_{z_{\beta}}^{\infty} N_{\Sigma_{A}} \left((z_{d} \, z_{r} \, z_{m} )^t \right) dz_{r} dz_{m} dz_{d} \nonumber \\
&\implies \int_{z_{\alpha}}^{\infty} \int_{z_{\gamma}}^{\infty} N_{\Sigma_{dm}} \left((z_{d} \, z_{m})^t\right) \left( \Phi \left(a_{1} z_{d} + b_{1} z_{m} + c_{1} \right) \right. \nonumber \\
&\phantom{\implies \int_{z_{\alpha}}^{\infty} \int_{z_{\gamma}}^{\infty} N_{\Sigma_{dm}} \left((z_{d} \, z_{m})^t\right) } \left. - \Phi \left(a_{0}z_{d} + b_{0} z_{m} + c_{0} \right) \right) dz_{d} dz_{m} = 0 \label{eq:intbeta}
\end{align}
where
\begin{align}
a_{0} &= -\frac{\rho_{dm}\rho_{rm}}{\sqrt{|\Sigma_{dm}||\Sigma_{A}|}} \nonumber \\
b_{0} &= \frac{\rho_{rm}}{{\sqrt{|\Sigma_{dm}||\Sigma_{A}|}}}\nonumber \\
c_{0} &= - \sqrt{\frac{|\Sigma_{dm}|}{|\Sigma_{A}|}}z_{\beta} \nonumber \\
a_{1} &= \frac{\rho_{ds} - \rho_{dm}\rho_{sm}}{\sqrt{|\Sigma_{dm}||\Sigma_{B}|}} \nonumber \\
b_{1} &= \frac{\rho_{sm} - \rho_{ds}\rho_{dm}}{\sqrt{|\Sigma_{dm}||\Sigma_{B}|}} \nonumber \\
c_{0} &= -\sqrt{\frac{|\Sigma_{dm}|}{|\Sigma_{B}|}}z_{\beta^*} \end{align}
The asymptotic property of $\beta^*$ follows from corollary~\ref{corollary:ylimit}.

If $\gamma=1$, we have from definition~\ref{eq:betastardef2} in the main paper
\begin{align}
&\int_{z_{\alpha}}^{\infty} \int_{z_{\beta^{*}}}^{\infty} \frac{1}{2 \pi \sqrt{1-\rho_{ds}^{2}}} exp \left( -\frac{1}{2(1-\rho_{ds}^{2})} \left(x^{2}+ y^{2} - 2 \rho x y \right) \right) dx dy \nonumber \\
&= \int_{z_{\alpha}}^{\infty} \int_{z_{\beta}}^{\infty} \frac{1}{2 \pi} exp \left( -\frac{1}{2} \left(x^{2}+ y^{2} \right) \right) dx dy\nonumber \\
&\implies \int_{z_{\alpha}}^{\infty} e^{-\frac{y^{2}}{2}} \Phi \left( \frac{z_{\beta^{*}}- \rho_{ds} y}{\sqrt{1-\rho_{ds}^{2}}} \right) dy = \int_{z_{\alpha}}^{\infty} e^{-\frac{y^{2}}{2}} \Phi (z_{\beta})
\end{align}
from which the result follows from an application of lemma~\ref{lemma:genlimit}.

\subsubsection{Size of $\beta$, $\beta^*$ and $\beta^\perp$}
\label{apx:betabound}

To show that $\beta^*<\beta$, we show that if we set $z_{\beta^{*}}=z_{\beta}$ in the integral~\ref{eq:intbeta}, then the integral is positive. Since it is decreasing with $z_{\beta}^*$ (as $\Phi$ is increasing) we must have $z_{\beta}^*>z_{\beta}$ if the integral is to be 0. A similar argument can be used to show that $\beta^\perp<\beta^*$. Denote by $I(z_{d},z_{m})$ the value of the integrand of ~\ref{eq:intbeta} with $z_{\beta^*}=z_{\beta}$.

Consider the line $a_{1} z_{d} + b_{1} z_{m} + c_{1} = a_{0}z_{d} + b_{0} z_{m} + c_{0}$ on the $(z_{d},z_{m})$ plane on which the integrand of~\ref{eq:intbeta} is 0. The gradient of this line is 
\begin{align}
\frac{a_{0}-a_{1}}{b_{0}-b_{1}} &= \frac{\sqrt{n_{0}^{i}(n_{0} + n_{0}') n_{1} (n_{0} + n_{1})}}{\sqrt{(n_{1} + n_{1}')(n_{0} + n_{0}' + n_{1} + n_{1}')}} \times \\
&\phantom{=} \frac{n_{0}' n_{1} (n_{0} + n_{0}' + n_{1} + n_{1}') - (n_{0} + n_{0}')|n_{0}' n_{1} - n_{0} n_{1}'|}{n_{0}' (n_{0} + n_{0}')  n_{1} (n_{0} + n_{1}) - (n_{0}^2 + n_{0} n_{0}' + n_{0}' n_{1})|n_{0}' n_{1} - n_{0} n_{1}'|}
\end{align}
Since $|n_{0}' n_{1} - n_{0} n_{1}'| \geq (n_{0}' n_{1} - n_{0} n_{1}')$ the numerator of the second fraction is greater than or equal to
\begin{align}
n_{0}' n_{1} (n_{0} + n_{0}' + n_{1} + n_{1}') - (n_{0} + n_{0}')(n_{0}' n_{1} - n_{0} n_{1}') &= n_{0}^2 n_{1}' + n_{0}' (n_{1}^2 + n_{0} n_{1}' + n_{1} n_{1}') \nonumber \\
&> 0
\end{align}
and similarly the denominator is greater than or equal to 
\begin{equation}
n_{0} (n_{0}^2 n_{1}' + n_{0}' (n_{1}^2 + n_{0} n_{1}' + n_{1} n_{1}'))> 0
\end{equation}
so the gradient is positive. If $b_{1}-b_{0}>0$, $I(z_{d},z_{m})$ is positive if $(z_{d},z_{m})$ falls above the line, and negative if below it; if $b_{1}-b_{0}<0$, the other way around. Assume for the moment that $b_{1}-b_{0}<0$. 

If the point $(z_{\alpha},z_{\gamma})$ lies above the line, then since $I(z_{d},z_{m})$ is negative in the region $(-\infty,z_{\alpha}) \times (z_{\gamma},\infty)$, we have  
\begin{align}
\int_{z_{\alpha}}^{\infty} \int_{z_{\gamma}}^{\infty} I(z_{d},z_{m}) dz_{d} dz_{m} &\geq \int_{z_{\alpha}}^{\infty} \int_{z_{\gamma}}^{\infty} I(z_{d},z_{m}) dz_{d} dz_{m} \nonumber \\
&\phantom{=} + \int_{-\infty}^{z_\alpha} \int_{z_{\gamma}}^{\infty} I(z_{d},z_{m}) dz_{d} dz_{m} \nonumber \\
&= \int_{\infty}^{\infty} \int_{z_{\gamma}}^{\infty} I(z_{d},z_{m}) dz_{d} dz_{m}
\end{align}
If the point lies below the line, let $z_{\gamma}' > z_{\gamma}$ be defined such that the point  $(z_{\alpha},z_{\gamma}')$ lies on the line. Since $I(z_{d},z_{m})$ is positive in the region $(z_{\alpha},\infty) \times (z_{\gamma},z_{\gamma}')$ and negative in the region $(-\infty,z_{\alpha}) \times (z_{\gamma},\infty)$, we have
\begin{align}
\int_{z_{\alpha}}^{\infty} \int_{z_{\gamma}}^{\infty} I(z_{d},z_{m}) dz_{d} dz_{m} &\geq \int_{z_{\alpha}}^{\infty} \int_{z_{\gamma}}^{\infty} I(z_{d},z_{m}) dz_{d} dz_{m} \nonumber \\
&\phantom{=} - \int_{z_{\alpha}}^{\infty} \int_{z_{\gamma}}^{z_{\gamma}'} I(z_{d},z_{m}) dz_{d} dz_{m} \nonumber \\
&\phantom{=} + \int_{-\infty}^{z_\alpha} \int_{z_{\gamma}'}^{\infty} I(z_{d},z_{m}) dz_{d} dz_{m} \nonumber \\
&= \int_{-\infty}^{\infty} \int_{z_{\gamma}'}^{\infty} I(z_{d},z_{m}) dz_{d} dz_{m}
\end{align}
so it is sufficient to prove that the integral is positive when the range $(z_{\alpha},\infty)$ is replaced with $(-\infty,\infty)$. Similar arguments can be used when $b_{1}-b_{0}>0$, in which case it is sufficient to prove positivity when $z_{\gamma}=0$. 

This enables $z_{d}$ (or $z_{m}$) to be integrated out, namely reducing to showing that
\begin{align}
\int_{z_{\beta}}^{\infty} \int_{z_{\gamma}}^{\infty} N_{\left( \begin{smallmatrix} 1 & \rho_{sm} \\ \rho_{sm} & 1 \end{smallmatrix} \right)} ((z_{s} \, z_{m})^t) - N_{\left( \begin{smallmatrix} 1 & \rho_{rm} \\ \rho_{rm} & 1 \end{smallmatrix} \right)} ((z_{s} \, z_{m})^t) dz_{m} dz_{s} &>0 \nonumber \\
\Leftrightarrow \int_{z_{\beta}}^{\infty} \frac{1}{2\pi}exp\left(\frac{1}{2}z_{s}^{2}\right) \left(\Phi \left(\frac{\rho_{sm}z_{s}-z_{\gamma}}{1-\rho_{sm}^2}\right)-\Phi \left(\frac{\rho_{rm}z_{s}-z_{\gamma}}{1-\rho_{rm}^2}\right)\right) &>0
\end{align}
The second part of the integrand is monotonically increasing in $z_{s}$ as $\rho_{sm}>\rho_{rm}$. Thus the integral is minimised as $z_{\beta} \to -\infty$, at which the value is $\Phi(z_{\gamma})$, which is positive. 

\subsection{SNPs with aberrant allele frequency in one group}

\subsubsection{$R_{B}<R_{A}$ for SNPs with aberrance in $C_{1}$}
\label{apx:abc1}
If SNPs have aberrant MAF in $C_{1}$ only, we have $E(z_{d})=\zeta_{d} \neq 0$, $E(z_{m})=\zeta_{m} \neq 0$ and $E(z_{s})=E(z_{r})=0$. As noted in the main text, as $\zeta_{d} \to 0$, $R_{B}, R_{A} \to P_{0}$ (equation~\ref{eq:betastardef} in the main paper) and 
\begin{align}
\lim_{\zeta_{d} \to \infty} R_{B} &= \lim_{\zeta_{d} \to \infty} \left( \int_{z_{\alpha}-\zeta_{d}}^{\infty} \int_{z_{\beta^{*}}}^{\infty} \int_{z_{\gamma}-\zeta_{m}}^{\infty} N_{\Sigma_{B}} \left( (z_{d} \, z_{s} \, z_{m} )^t\right) dz_{s} dz_{m} dz_{d} \right. \nonumber \\
&\phantom{= \lim_{\zeta_{d} \to \infty}} \left. + \int_{z_{\alpha}+\zeta_{d}}^{\infty} \int_{z_{\beta^{*}}}^{\infty} \int_{z_{\gamma}+\zeta_{m}}^{\infty} N_{\Sigma_{B}} \left( (z_{d} \, z_{s} \, z_{m} )^t\right) dz_{s} dz_{m} dz_{d} \right) \nonumber \\
&= \Phi(-z_{\beta}^*) = \frac{\beta^*}{2}
\end{align}
and similarly, $R_{A} \to \frac{\beta}{2}$, $R_{B} \to \frac{\beta^*}{2}$ as $\zeta_{d} \to \pm \infty$, with $\beta^* < \beta$ as shown above. For $\beta^*$ defined by~\ref{eq:betastardef2} in the main paper, we show here that $R_{A}>R_{B}$ for all $\zeta_{d}$. For the more general definition of $\beta^*$ (equation~\ref{eq:betastardef} in the main paper), the inequality $R_{B}<R_{A}$ may not hold for all $\zeta_{d}$. However, in practice, the inequality holds for almost all $\zeta_{d}$ and any deviation is small and near $\zeta_{d}=0$.

Define the shorthand $N_{\rho}(x,y)$ as the value at $(x,y)$ of the bivariate normal PDF with mean $\left( \begin{smallmatrix} 0 \\ 0 \end{smallmatrix} \right)$ and variance $\left( \begin{smallmatrix} 1 & \rho \\ \rho & 1 \end{smallmatrix} \right)$, and $erfc(x) = 2 \left( 1- \Phi(\sqrt{2}x) \right)$ as the complementary error function. In this section, $\rho$ refers to $\rho_{ds}$.

Consider the value $R_{A}-R_{B}$ as a function of $\zeta_{d}$. We will show that the partial derivative $\frac{\delta}{\delta \zeta_{d}} \left(R_{A} - R_{B} \right)$ is zero only when $\zeta_{d}=0$. Because $R_{A}-R_{B}$ passes through the origin and is asymptotically positive in both directions, it is positive for all $\zeta_{d} \neq 0$. We have
\begin{align}
R_{A} - R_{B} &= \left( \int_{z_{\beta}}^{\infty} \int_{z_{\alpha}-\zeta_{d}}^{\infty}N_{0}(x,y) dx dy - \int_{z_{\beta^{*}}}^{\infty} \int_{z_{\alpha}-\zeta_{d}}^{\infty}N_{\rho}(x,y) dx dy \right) \nonumber \\
&\phantom{{} = +} + \left( \int_{-\infty}^{-z_{\beta}} \int_{-\infty}^{-z_{\alpha}+\zeta_{d}} N_{0}(x,y) dx dy - \int_{-\infty}^{-z_{\beta^{*}}} \int_{-\infty}^{-z_{\alpha}+\zeta_{d}} N_{\rho}(x,y) dx dy \right) \\
\frac{\delta}{\delta \zeta_{d}} \left(R_{A} - R_{B} \right) &= \left( \int_{z_{\beta}}^{\infty} \frac{\delta}{\delta \zeta_{d}} \int_{z_{\alpha}-\zeta_{d}}^{\infty}N_{0}(x,y) dx dy - \int_{z_{\beta^{*}}}^{\infty} \frac{\delta}{\delta \zeta_{d}} \int_{z_{\alpha}-\zeta_{d}}^{\infty}N_{\rho}(x,y) dx dy \right) \nonumber \\
&\phantom{{} = +} + \left( \int_{z_{\beta}}^{\infty} \frac{\delta}{\delta \zeta_{d}} \int_{z_{\alpha}+\zeta_{d}}^{\infty}N_{0}(x,y) dx dy - \int_{z_{\beta^{*}}}^{\infty} \frac{\delta}{\delta \zeta_{d}} \int_{z_{\alpha}+\zeta_{d}}^{\infty}N_{\rho}(x,y) dx dy \right) \nonumber \\
&= \frac{1}{2\sqrt{2\pi}} erfc \left(\frac{z_{\beta}}{\sqrt{2}}\right) \left(e^{-\frac{1}{2}(\zeta_{d}-z_{\alpha})^{2}} - e^{-\frac{1}{2}(\zeta_{d}+z_{\alpha})^{2}}\right) \nonumber 
\end{align}
\vspace{-1.5cm}
\begin{align}
&- \frac{1}{2\sqrt{2\pi}} \left( e^{-\frac{1}{2}(\zeta_{d}-z_{\alpha})^{2}}erfc \left(\frac{z_{\beta^{*}} + \rho(\zeta_{d}-z_{\alpha})}{\sqrt{2(1-\rho^{2})}}\right) - e^{-\frac{1}{2}(\zeta_{d}+z_{\alpha})^{2}} erfc \left(\frac{z_{\beta^{*}} - \rho(\zeta_{d}+z_{\alpha})}{\sqrt{2(1-\rho^{2})}}\right) \right) \nonumber
\end{align}
Showing that $\frac{\delta}{\delta \zeta_{d}} \left(R_{A} - R_{B} \right) >0$ when $\zeta_{d}>0$ is equivalent to showing that $(a-b)-(pa-qb) > 0$ where $a=
e^{-\frac{1}{2}(\zeta_{d}-z_{\alpha})^{2}}$, $b=e^{-\frac{1}{2}(\zeta_{d}+z_{\alpha})^{2}}$, $p=\frac{erfc \left(\frac{z_{\beta^{*}} + \rho(\zeta_{d}-z_{\alpha})}{\sqrt{2(1-\rho^{2})}}\right)}{erfc \left(\frac{z_{\beta}}{2}\right)}$ and $q=\frac{erfc \left(\frac{z_{\beta^{*}} + \rho(\zeta_{d}-z_{\alpha})}{\sqrt{2(1-\rho^{2})}}\right)}{erfc \left(\frac{z_{\beta}}{2}\right)}$.

Since $(\zeta_{d}-z_{a})^{2}<(\zeta_{d}+z_{a})^{2}$ for $\zeta_{d}>0$, we have $a>b$. Because $erfc$ is strictly decreasing, we have $p<q$. Because $\frac{\delta p}{\delta \zeta_{d}} <0$, we have
\begin{equation}
p < \frac{erfc \left(\frac{z_{\beta^{*}} -z_{\alpha})}{\sqrt{2(1-\rho^{2})}}\right)}{erfc \left(\frac{z_{\beta}}{2}\right)} < 1
\end{equation}
where the second inequality arises because $z_{\beta^{*}}>
\sqrt{1-\rho^{2}}z_{\beta} + \rho z_{\alpha}$. Thus $pa-qb<pa-pb=p(a-b)<a-b$, and $\frac{\delta}{\delta \zeta_{d}} \left(R_{A} - R_{B} \right) >0$ as required.


\subsection{Upper bound on $R_{B}-R_{A}$ with aberrance in $C_{1}'$}
\label{apx:c1vaberrance}

For SNPs with aberrance in $C_{1}'$, we have $E(z_{d})=0$, $E(z_{r})=\zeta_{r} \neq 0$, $E(z_{s})=\zeta_{s} \neq 0$ and $E(z_{m})=\zeta_{m} \neq 0$. As above $R_{A},R_{B} \to P_{0}$ as $\zeta_{r} \to 0$, and by similar arguments to the section above, $R_{A},R_{B} \to \frac{\alpha}{2}$ as $\zeta_{r} \to \pm \infty$. 

For $\beta^*$ defined as per equation~\ref{eq:betastardef2} in the main paper, it is possible to derive an approximate (asymptotically accurate) upper bound on $R_{B}-R_{A}$, corresponding to the most serious increase in FPR. The approach is not readily applied to the general definition of $\beta^*$, but again the difference is typically small in practice. 

To first order
\begin{equation}
\frac{\zeta_{s}}{\zeta_{r}} = \frac{\left(\frac{\mu_{1}'-\mu_{0}'}{SE(m_{1}'-m_{0}')} \right)}{\left(\frac{\mu_{1}'-\mu_{0}}{SE \left( m_{1}'-\frac{m_{0}n_{0} + m_{0}'n_{0}'}{n_{0}+n_{0}'} \right) } \right)} = \sqrt{\frac{(n_{0}+n_{0}')(n_{0}'+n_{1}'}{n_0'(n_0+n_0'+n_1')}} \myeq k \label{eq:k1def} 
\end{equation}
%
%
%
%
Now
\begin{align}
R_{B} - R_{A} &= \left( \int_{z_{\beta^{*}}- \zeta_{s}}^{\infty} \int_{z_{\alpha}}^{\infty}N_{\rho}(x,y) dx dy - \int_{z_{\beta}-\zeta_{r}}^{\infty} \int_{z_{\alpha}}^{\infty}N_{0}(x,y) dx dy \right) \nonumber \\
&\phantom{{} = +} + \left( \int_{z_{\beta^{*}}+ \zeta_{s}}^{\infty} \int_{z_{\alpha}}^{\infty}N_{\rho}(x,y) dx dy - \int_{z_{\beta}+\zeta_{r}}^{\infty} \int_{z_{\alpha}}^{\infty}N_{0}(x,y) dx dy \right)
\end{align}
Define $z_{r}^+$, $z_{r}^-$ such that
\begin{align}
\int_{z_{r}^-}^{\infty} \int_{z_{\alpha}}^{\infty}N_{0}(x,y) dx dy &= \int_{z_{\beta^{*}}- \zeta_{s}}^{\infty} \int_{z_{\alpha}}^{\infty}N_{\rho}(x,y) dx dy \nonumber \\
\int_{z_{r}^+}^{\infty} \int_{z_{\alpha}}^{\infty}N_{0}(x,y) dx dy &= \int_{z_{\beta^{*}}+ \zeta_{s}}^{\infty} \int_{z_{\alpha}}^{\infty}N_{\rho}(x,y) dx dy
\end{align}
From equation~\ref{eq:infbetastar2} in Appendix~\ref{apx:infbetastar}, we have $z_{\beta^*}-\zeta_{s} \approx \sqrt{1-\rho^2}z_{r}^{-} - \rho z_{\alpha}$ and $z_{\beta^*}+\zeta_{s} \approx \sqrt{1-\rho^2}z_{r}^{+} - \rho z_{\alpha}$.

Noting that $\int_{a}^\infty \int_{b}^{\infty} N_{0}(x,y) dx dy = \Phi(-a) \Phi(-b)$ and $\Phi(x) = 1-\Phi(-x)$ we now have
\begin{align}
R_{B} - R_{A} &= 
\Phi(-z_{\alpha}) \left( \Phi(z_{\beta} - \zeta_{r}) - \Phi(z_{r}^-) +  \Phi(z_{\beta}+\zeta_{r}) - \Phi(z_{r}^+) \right)
\end{align}
Applying the identity $\Phi(-z_{\alpha})=\frac{\alpha}{2}$ and approximations $z_{\beta}^* \approx \sqrt{1-\rho^{2}}z_{\beta} + \rho z_{\alpha}$, $\zeta_{s}\approx k \zeta_{r}$, yields
\begin{align}
z_{r}^{-} &\approx \frac{z_{\beta^*}- \zeta_{s}+ \rho z_{\alpha}}{\sqrt{1-\rho^2}} \approx z_{\beta} - \frac{k}{\sqrt{1-\rho^{2}}} \zeta_{r} \nonumber \\
z_{r}^{+} &\approx z_{\beta} + \frac{k_{1}}{\sqrt{1-\rho^{2}}} z_{0}' 
\end{align}
\begin{equation}
R_{B}-R_{A} \approx \frac{\alpha}{2} \left( \Phi \left( z_{\beta} - \frac{k}{\sqrt{1-\rho^{2}}} \zeta_{r} \right)-\Phi(z_{\beta}-\zeta_{r}) +  \Phi \left( z_{\beta} + \frac{k}{\sqrt{1-\rho^{2}}} \zeta_{r} \right)-\Phi(z_{\beta}+ \zeta_{r}) \right)
\end{equation}
Considered as a function of $\zeta_{r}$, the value $R_{B}-R_{A}$ will be 0 at $\zeta_{r}=0$ and tend to 0 as $\zeta_{r} \to \pm \infty$. It will be maximised approximately at the points where $\Phi(z_{\beta}-\zeta_{r})$ or $\Phi(z_{\beta}+\zeta_{r})$ are changing most rapidly; that is, $\zeta_{r}=\pm z_{\beta}$. At $\zeta_{r}= z_{\beta}$, the contribution to the value $R_{B}-R_{A}$ from the difference $\Phi \left( z_{\beta} + \frac{k}{\sqrt{1-\rho^{2}}} \zeta_{r} \right)-\Phi(z_{\beta}+\zeta_{r})$ is negligible (and similarly for the other difference when $\zeta_{r}= -z_{\beta}$). Using the first-order approximation for $\Phi(z_{\beta}-\zeta_{r})$ about $\zeta_{r}=z_{\beta}$ yields
\begin{equation}
max \left(R_{B}-R_{A} \right) \approx \frac{\alpha}{2 \sqrt{2 \pi}} \left(\frac{k}{\sqrt{1-\rho^{2}}} - 1 \right) z_{\beta}
\end{equation}
In general, this value is substantially less than $\alpha$. 

All instances of `approximately equal' are asymptotic limits as $z_{a} \to \infty$ and $n_{0}$, $n_{0'}$, $n_{1}$, $n_{1}' \to \infty$ with $z_{0}'$ held finite. 

\subsection{Aberrance in $C_{0}'$}

For SNPs aberrant in $C_{0}'$, again $E(z_{d})=0$, $E(z_{r})=\zeta_{r} \neq 0$, $E(z_{s})=\zeta_{s} \neq 0$ and $E(z_{m})=\zeta_{m} \neq 0$. As above $R_{A},R_{B} \to P_{0}$ as $\zeta_{r} \to 0$, and $R_{A},R_{B} \to \frac{\alpha}{2}$ as $\zeta_{r} \to \pm \infty$. In this case, $R_{B}$ is typically less than $R_{A}$. 

\subsection{General aberrance in replication cohorts}

For $\beta^*$ defined according to \ref{eq:betastardef2} in the main paper, the increase in FPR $R_{B}-R_{A}$ for method B for a SNP with aberrance in $C_{1}'$ is generally smaller than the decrease in FPR $R_{A}-R_{B}$ for a SNP with a similarly-sized aberrance in $C_{0}'$, in that the integral of the difference over the range of $\zeta_{r}$ is generally smaller in the former.

We define $k$ as in the section above and $k_{1}=\frac{\zeta_{s}}{\zeta_{r}} \big| C_{0}' \textrm{ aberrant} = \sqrt{\frac{n_{0}'(n_{0}' + n_{1}')}{(n_{0}+n_{0}')(n_{0}+n_{0}'+n_{1}')}}$ similarly. 
%
%
%
%
Now for $C_{0}'$ aberrant
\begin{equation}
R_{A}-R_{B} \approx \frac{\alpha}{2} \left( \Phi(z_{\beta}-\zeta_{r}) - \Phi \left( z_{\beta} - \frac{k_{1}}{\sqrt{1-\rho^{2}}} \zeta_{r} \right) + \Phi(z_{\beta}+\zeta_{r}) - \Phi \left( z_{\beta} + \frac{k_{1}}{\sqrt{1-\rho^{2}}} \zeta_{r} \right) \right)
\end{equation}
%
Since $\int_{0}^{x} \Phi(z) dz =x \Phi(x) + \frac{1}{\sqrt{2 \pi}} \left( e^{-\frac{x^{2}}{2}} - 1 \right)$, we have 
\begin{align}
\int_{0}^{\infty} \left( \Phi(h-z) - \Phi(h - k z) \right) dz 
&= \left(1-\frac{1}{k} \right) \left( \frac{1}{\sqrt{2 \pi}} e^{-\frac{1}{2} h^{2}} + h \Phi(h) \right)
 \\
\int_{0}^{\infty} \left( \Phi(h+z) - \Phi(h + k z) \right) dz &= \left(1-\frac{1}{k} \right) \left( -\frac{1}{\sqrt{2 \pi}} e^{-\frac{1}{2} h^{2}} + h \Phi(-h) \right)
\end{align}
%
Thus with aberrant $C_{0}'$
\begin{equation}
\int_{0}^{\infty} \left( R_{A}-R_{B}  \right) d\zeta_{r} = \frac{\alpha}{2}  \left(1-\frac{\sqrt{1-\rho^{2}}}{k_{1}} \right) z_{\beta}
\end{equation}
Comparing $R_{A}$ and $R_{B}$ under the two aberrance scenarios with the same $\zeta_{d}$
\begin{equation}
\frac{\int_{0}^{\infty} \left( R_{A}-R_{B}  \right) d\zeta_{d} \, \, \textrm{[$C_{0}'$ aberrant]}}{\int_{0}^{\infty} \left( R_{B}-R_{A}  \right) d\zeta_{d} \, \, \textrm{[$C_{1}'$ aberrant]}} = \frac{1-\frac{\sqrt{1-\rho^{2}}}{k_{1}}}{\frac{\sqrt{1-\rho^{2}}}{k}-1}
\end{equation}
For this to be $>1$, a necessary condition is 
%
$\left(1-\frac{\sqrt{1-\rho^{2}}}{k_{2}}\right) > \left(\frac{\sqrt{1-\rho^{2}}}{k_{1}}-1 \right)$
%
From the definitions of $\rho_{ds}$ (Appendix~\ref{apx:zcor}), $k$ (equation~\ref{eq:k1def}) and $k_{1}$, this is equivalent to
\begin{align}
\sqrt{\frac{n_{0}+n_{0}'+n_{1}'}{n_{0}'+n_{1}'}}\sqrt{1-\frac{n_{0}n_{1}n_{1}'}{(n_{0}+n_{0}')(n_{0}+n_{1})(n_{0}+n_{0}'+n_{1}')}} \left(\sqrt{\frac{n_{0}'}{n_{0}+n_{0}'}} + \sqrt{\frac{n_{0}+n_{0}'}{n_{0}}} \right) > 2 \nonumber
\end{align}
The final term in this product is of the form $x + \frac{1}{x}$ so is greater than 2. A sufficient condition is thus
\begin{align}
\frac{n_{0}+n_{0}'+n_{1}'}{n_{0}'+n_{1}'}\left( 1-\frac{n_{0}n_{1}n_{1}'}{(n_{0}+n_{0}')(n_{0}+n_{1})(n_{0}+n_{0}'+n_{1}')} \right) &\geq 1 \nonumber \\
\Longleftrightarrow n_{0}^{2} + n_{0}(n_{0}'+n_{1}) + n_{1}(n_{0}'-n_{1}') &\geq 0
\end{align}
which holds in most study designs.

\end{document}